\documentclass[12pt,ptreqno]{amsart}
\usepackage{amsmath,amssymb,amsfonts,amsthm,eucal,amscd} 
\usepackage{hyperref}
\usepackage[margin=1in]{geometry}

\usepackage[T1]{fontenc}
\usepackage{tgtermes}
\usepackage[american]{babel}
\usepackage{mathrsfs}
\usepackage{graphicx}
\usepackage[all]{xy}
\usepackage{subfigure}
\usepackage[usenames,dvipsnames,svgnames,table]{xcolor}
\usepackage{tikz}
\usetikzlibrary{matrix,arrows,backgrounds,calc,chains,automata,positioning,patterns}
\usetikzlibrary{shapes.geometric,calc}
\usepackage{enumerate}
\usepackage{setspace}
\usepackage{hyperref}

\begin{document}

\numberwithin{equation}{section}

\newtheorem*{theorem*}{Theorem} 
\newtheorem{theorem}{Theorem}[section] 
\newtheorem{lemma}[theorem] {Lemma}
\newtheorem{proposition}[theorem] {Proposition} 
\newtheorem{corollary}[theorem] {Corollary} 
\theoremstyle{definition} 
\newtheorem{definition}[theorem] {Definition} 
\newtheorem{example}[theorem] {Example} 
\newtheorem{notation}[theorem] {Notation} 
\newtheorem{conjecture}[theorem] {Conjecture}
\newtheorem{remark}[theorem] {Remark} 
\theoremstyle{remark} 

\makeatletter
\@addtoreset{figure}{section}
\makeatother
\renewcommand{\thefigure}{\arabic{section}.\arabic{figure}}

\def\Rem#1{\noindent Remark~\ref{#1}}
\def\Def#1{\noindent Definition~\ref{#1}}
\def\Lem#1{\noindent Lemma~\ref{#1}} 
\def\Thm#1{\noindent Theorem~\ref{#1}} 
\def\Sec#1{\noindent Section~\ref{#1}} 
\def\Fig#1{\noindent Figure~\ref{#1}} 

\def\Prop{\noindent Proposition}
\def\Ex{\noindent Example}
\def\Cor{\noindent Corollary}

\long\def\BTHMm#1#2{\begin{theorem*}[#1]#2\end{theorem*}}
\long\def\BTHM#1#2{\begin{theorem}\LL{#1}#2\end{theorem}}
\long\def\BDF#1#2{\begin{definition}\LL{#1}#2\end{definition}}
\long\def\BPROP#1#2{\begin{proposition}\LL{#1}#2\end{proposition}}
\long\def\BLEM#1#2{\begin{lemma}\LL{#1}#2\end{lemma}}
\long\def\BNT#1#2{\begin{notation}\LL{#1}#2\end{notation}}
\long\def\BREM#1#2{\begin{remark}\LL{#1}#2\end{remark}}
\long\def\BEX#1#2#3{\begin{example}[#2]\LL{#1}#3\end{example}}
\long\def\BCOR#1{\begin{corollary}#1\end{corollary}}

\long\def\BEQ#1#2{\begin{equation}\LL{#1}#2\end{equation}}
\long\def\BB#1#2#3{\begin{#1}\LL{#2}#3\end{#1}} 
\long\def\BBF#1#2#3{\begin{figure}[h]#3\caption{#2}\LL{#1}\end{figure}} 

\def\Pr{\begin{proof}}
\def\rP{\end{proof}} 
\def\bydef#1{{\itshape\bfseries\rmfamily#1\,}}
\def\bydef#1{#1}

\def\i{iii}
\def\p{\mathcal P}
\def\X{X}
\def\A{\mathcal A}
\def\E{\mathcal E}
\def\J{\mathcal J}
\def\dd{\mathcal D}
\def\varpsi{\varPsi}
\def\uv{U}
\def\ud{D}
\def\udd{D^*}
\def\h{\mathcal H}
\def\L{\mathcal L}
\def\C{\mathbb C}
\def\R{\mathbb R}
\def\fone{\xi_R^{\phantom{y}}}
\def\ftwo#1{\psi_{R,#1}^{\phantom{y}}}
\def\zetaa#1#2{\zeta_{#2,#1}^{\phantom{y}}}

\def\bb#1#2{\begin{#1}\LL{#2}} 
\def\B#1{$\discretionary{}{}{}$#1$\discretionary{}{}{}$}
\newcounter{tmp}
\def\blist{\begin{list}{\hskip-19pt(\arabic{tmp})}{\usecounter{tmp}}}
\def\sitem{\parskip\par\item}

\def\boxx#1{\leavevmode\vbox{\hbox to 0pt{\hss\raise1.8ex\vbox 
to 0pt{\vss\hrule\hbox{\vrule\kern.75pt\vbox{\kern.75pt\hbox{\tiny #1}\kern.75pt}\kern.75pt\vrule}\hrule}}}\relax} 

\def\LL#1{\label{#1}\protect\boxx{#1}} 
\def\LL#1{\label{#1}}

\def\e#1{(\reF{#1})}
\def\reF#1{\ref{#1}}

\newenvironment{procLaim}{\it}{\par\smallskip}
\def\proclaim#1{\par\smallskip\noindent {#1}\bb{procLaim}{?}}
\def\endproclaim{\end{procLaim}} 

\def\zf{zeta function } 
\def\szf{spectral zeta function } 
\def\zfoap{zeta function of a polynomial} 
\def\zfotp{zeta function of the polynomial } 
\def\Str{Strichartz} 
\def\Edef{\overset{\textup{\textsf{\tiny def}}}=}
\def\Si{Sier\-pi\'n\-ski} 
\def\Sig{Sier\-pi\'n\-ski gasket} 
\def\Sil{Sierpi\'nski lattice}
\def\iSig{infinite Sierpi\'nski gasket} 
\def\Sip{Sierpi\'nski pre-gasket}
\def\Lp{Laplacian}
\def\bfit#1{{\bfseries\itshape#1\,}}

\def\<{\langle}
\def\>{\rangle}

\def\iiii{}
\def\iii{}
\def\ii{}
\def\ih{}
\def\ps{}
\def\pms{}
\def\lala{}
\def\ms{}
\def\sm#1{\mbox{$#1$}}\def\sm{}
\def\lg#1{\mbox{$#1$}}
\def\Lg#1{\mbox{$#1$}}
\def\LG#1{\mbox{$#1$}}

\def\split{}
\def\endsplit{}
\def\ndt{\noindent}\def\ndt{}
\def\np{\newpage}
\def\pb{\pagebreak}

\long\def\B{\par}

\def\<{\langle} 
\def\>{\rangle} 
\def\({\<\<} 
\def\){\>\>}
\def\tri{| \hskip-0.02in |\hskip-0.02in |}
\def\e{\alpha}
\def\p{\mathcal P}
\def\E{\mathcal E}
\def\J{\mathcal J}
\def\dd{\mathcal D}\def\D{\mathcal D}
\def\ddD#1{#1,#1}
\def\h{\mathcal H}
\def\Z{\mathcal Z}
\def\J{\mathcal J}
\def\A{\mathcal A}
\def\L{\mathcal Z}
\def\M{\mathcal M}
\def\N{\mathcal N}
\def\C{\mathbb C}
\def\R{\mathbb R}

\def\Edef{\overset{\textup{\textsf{ def}}}=}
\def\Si{{Sierpi\'nski}} 
\def\Sig{Sierpi\'nski gasket} 
\def\Sil{Sierpi\'nski lattice}
\def\iSig{infinite Sierpi\'nski gasket} 
\def\Sip{Sierpi\'nski pre-gasket}
\def\Lp{Laplacian}
\def\uv{U}
\def\ud{D}
\def\udd{D^ *}
\def\X{
X}
\def\bX{\bar 
X}

\def\dim{\textup{\textsf{dim}}}
\def\supp{\textup{\textsf{supp}}}
\def\mult{\textup{\textsf{mult}}}
\def\deg#1{\textup{\textsf{deg}}_\ii {\raise-.2ex\hbox{$_\ii #1$}}}
\def\degg{\textup{\textsf{deg}}}
\def\ind#1{_\ii {\raise-.1ex\hbox{$_\ii #1$}}}
\def\inds#1{
{\raise-.25ex
\hbox{$
_\ii 
#1$}}}

\def\z{^\ii {\text{\bfseries (0)}}}
\def\ellz{\ell^\ii {2\text{$(0)$}}}
\def\C{{\bf C}_0} 
\def\D{{\bf D}_0} 
\def\Di{{\bf D}_1} 

\def\1{\mbox{$1$%
\hskip-.38em\rule{.44em}{.05ex}%
\hskip-.18em\rule{.06em}{1.5ex}%
\hskip-.16em{\raise1.48ex\hbox{\rule{.27em}{.05ex}}}\hskip0.1em%
}}

\long\def\NEWSECTION#1#2#3{\section{{#2}}
\B\centerline{#3}\B\B}
\long\def\NEWSECTION#1#2#3{\section{{#2}}\LL{#3}}

\title[Wave equation on fractals]
{Wave equation on one-dimensional fractals with spectral decimation and the complex dynamics of polynomials} 

\author[Andrews]{Ulysses Andrews}
\email{ulysses.andrews@uconn.edu}
\author[Bonik]{Grigory Bonik}
\email{gregory@bonik.org}
\author[Chen]{Joe P. Chen}
\email{joe.p.chen@uconn.edu$^1$, jpchen@colgate.edu}
\author[Martin]{Richard W. Martin}
\email{richard.w.martin@uconn.edu$^2$}
\author[Teplyaev]{Alexander Teplyaev}
\email{teplyaev@member.ams.org}
\thanks{Research supported in part by the NSF Grants DMS-1106982 
	and DMS-1262929.}
\address{\noindent Department of Mathematics, 
University of Connecticut, Storrs, CT 06269, USA.}
\address{\noindent $^1$\emph{Current address:} 
	Mathematics Department,
	Colgate University,	
    Hamilton, NY 13346, USA.}
\address{\noindent $^2$\emph{Current address:} Mathematics Department,
SUNY Broome Community College, Binghamton, NY 13902,~USA.}
\date{\today}
\begin{abstract}
We study the wave equation on one-dimensional self-similar fractal structures that 
can be analyzed by the spectral decimation method. We develop efficient 
numerical approximation techniques and also provide uniform estimates obtained by analytical methods. 
{\tableofcontents}
\end{abstract} 

\keywords{Wave equation, fractals, spectral decimation, infinite speed of propagation, time change.}
\subjclass[2010]{35L05, 42C15, 58J45 (primary); 
28A80, 
35P10, 
35P15, 
60J25, 
60J45,  
81Q35 (secondary)} 

\maketitle

\section{Introduction}\label{section1intro} 
The purpose of this paper is to study, 
both analytically and numerically, 
the wave equation on the unit interval 
endowed with a self-similar fractal measure. 
Previous studies of wave equation on fractals,
including numerical approximations, 
were published in 
\cite{CCNT,CDS,CSW,DSV,GRS,St1}. All these works have some, although not direct, relation to the classical paper \cite{St}, 
but are more directly related to 
 the fractal Fourier analysis, see \cite{Str03,Str-gaps}. 
 Our computational methods mostly come from the theoretical papers \cite{eigen1,eigen2,T,T04} 
that develop so-called spectral decimation method in the form 
applicable for to  numerical analysis. 

In general, there is a large 
literature dealing with analysis and probability on fractals in mathematical terms, such as  \cite[and references therein]{ASST,Ba,BBKT,BGH,BST,BNT,DRS,FS,Go,IPRRS,Ki1,Ki,LaL,KL,KL2,MT,merozeta,S,Sh}, 
and also extensive  mathematical 
physics literature, including  \cite{ADT1,ADT2,AH,AO,CT,CMT,D,DABK,Go,R,RT}. 
Of particular interest are the works 
studying the appearance of fractals in quantum gravity, including 
\cite{AJL2,AJL3,ACOS,Ca,Ca3,CM,CPR2,Englert,LR7,MPM,RS1,ReS}.

We consider a situation in which a good enough (fractal) Laplacian $\Delta$ is defined 
on $L^2(K,\mu)$, where 
a compact set $K$ (the unit interval in our case) equipped with a (fractal) Borel measure~$\mu$.
This Laplacian $\Delta$ 
is
 a point-wise limit or as the generator of a Kigami's resistance form (see Proposition~\ref{prop16}), 
and one can extend some of the classical numerical techniques to approximate some (intrinsically smooth) solutions of the wave equation initial value problem 
\begin{equation}
	\label{eq:waveeqn}
	\left\{ \begin{array}{ll} \partial_{tt} u = -\Delta u & \text{on}~ K \times [0,T],\\
		u(\cdot,0) = \phi &\text{on}~K,\\
		\partial_t u(\cdot,0)=\psi & \text{on} ~K.
	\end{array}\right.
\end{equation}
As is well known, if the spectrum of the Laplacian is discrete, then the solution of the wave equation can be represented in terms of $L^2(\mu)$-eigensolutions $\{\lambda_k, f_k\}_{k=0}^\infty$ of the Laplacian $\Delta$, with $\lambda_0 \lneq \lambda_1 \leq \lambda_2 \leq \cdots \uparrow \infty$ and $\Delta f_k = \lambda_k f_k$. Writing
\begin{equation}
	\phi = \sum_{k=0}^\infty \alpha_k f_k \quad \text{and} \quad \psi=\sum_{k=0}^\infty \beta_k f_k, 
\end{equation}
where $\alpha_k = \langle \phi, f_k\rangle_{L^2}$ and $\beta_k = \langle \psi, f_k\rangle_{L^2}$, one finds that $u$ admits the series representation
\begin{equation}
	u(x,t) = \sum_{k=0}^\infty \alpha_k f_k(x) \cos\left(t\sqrt{\lambda_k}\right) + \sum_{k=k_{\min}}^\infty \frac{\beta_k}{\sqrt{\lambda_k}} f_k(x) \sin\left(t\sqrt{\lambda_k}\right),
\end{equation}
where $k_{\min} := \min\{j\in \mathbb{N}\cup\{0\}: \lambda_j >0\}$. It is known that the series point-wise converges poorly and the numerical approximations are very unstable unless the smoothness of solutions can be controlled. 

In our setup, $K=I$ and $\mu$ is the fractal measure defined in Section \ref{section2eigenvalues}. For simplicity we assume that the initial velocity $\psi \equiv 0$, and so  the solution to (\ref{eq:waveeqn}) is
\begin{equation}
	u(x,t) = \sum_{k=0}^\infty \alpha_k f_k(x) \cos\left(t\sqrt{\lambda_k}\right).
\end{equation}
If we theoretically  assume that $\phi$ is given by a $\delta$-impulse at point $0$,  $\phi=\delta_0$, then we have that 
  $\alpha_k := \int_I\, f_k(x) \delta_0(x) \, \mu(dx)= f_k(0)$. 
Note that $\delta_0(x)$ is not a function by the unit atomic 
measure 
at zero, and so  the integral in this definition of $\alpha_k$ is to be understood as a formal expression, as in the theory of distributions (for the classical version, see \cite{StDistr}, and for the fractal version, see \cite{RoStr}).  This approach on a fractal space does not allow an accurate numerical approximation of the solutions. 

Therefore we concentrate on a situation where the initial condition is highly localize function, but is smooth in intrinsic sense, and we can show that the approximating series converges uniformly. This is an illustration of the general principle of Stricharz \cite{Str-gaps}: \emph{Laplacians on fractals with spectral gaps have nicer Fourier series}. However, the abstract result 
\cite{Str-gaps} does not include the estimate of the remainder which we obtain in our work. 

Numerically, we can only compute the eigensolutions of the fractal Laplacian up to a finite level, so in practice we solve the ``approximate'' wave equation
\begin{equation}
\label{eq:approxwaveeqn-}
\left\{ \begin{array}{ll} \partial_{tt} u_n = -\Delta_n u_n & \text{on}~ V_n \times [0,T],\\
u_n(\cdot,0) = \delta_0^{(n_0,n)} &\text{on}~V_n,\\
\partial_t u_n(\cdot,0)=0 & \text{on} ~V_n,
\end{array}\right.
\end{equation}
where $\displaystyle \delta_0^{(n,n_0)}=\sum_{k=0}^{3^{n_0}} \alpha_k f_{n,k}$ is the approximate $\delta$-function built up from the first $|V_{n_0}|=(3^{n_0}+1)$ eigenfunctions of $\Delta_n$ (with $\Delta_n f_{n,k} = \lambda_{n,k} f_{n,k}$), and $\alpha_k := \alpha_{n_0,k}\geq 0$ are the coefficients found in Section \ref{sec:specdecom}. Throughout the section $n_0$ will be fixed, and we will not mention $n_0$ explicitly unless the context demands it.
The solution to (\ref{eq:approxwaveeqn-}) has the series representation
\begin{equation}
\label{eq:unseries-}
u_n(x,t) = \sum_{k=0}^{3^{n_0}} \alpha_k f_{n,k}(x) \cos\left(t\sqrt{\lambda_{n,k}}\right) \quad\text{for all}~x\in V_n ~\text{and}~t\in [0,T].
\end{equation}
For each $t$, we harmonically extend the function $x\mapsto u_n(x, t)$ from $V_n$ to $I$. This procedure allows us to compare $u_n(x,t)$ with
\begin{equation}
\label{eq:useries-}
\tilde{u}(x,t) = \sum_{k=0}^{3^{n_0}} \alpha_k f_k(x) \cos\left(t\sqrt{\lambda_k}\right) \quad\text{for all}~ x\in I ~\text{and}~t\in [0,T],
\end{equation}
the solution of the wave equation on $(I,\mu)$ whose initial condition is the truncated series representation of the $\delta$-impulse.
We note that $\tilde{u}$ is 
differentiable in $t$ and continuous in $x$. 
However it is highly localized function 
at $t=0$, and therefore it 
mimics wave propagation from a delta function initial values.

Our paper is organized as follows. 
Section \ref{section2eigenvalues} contains the construction of the unit interval 
as a p.c.f. fractal, definition of the Dirichlet energy form, the definition of the corresponding Laplacian and its associated eigenvalues.  In Section \ref{section3eigenfunctions} we use spectral decimation to construct the eigenfunctions of the discrete Laplacian and prove that their limit is continuous.  The section concludes with the spectral decomposition of the delta function.  Section \ref{section5estimates} contains various technical estimates needed to show the convergence of solutions of the wave equation.  In Section \ref{sectionwaveequation} we give theoretical bounds on the approximations to the wave equations solutions and convergence information.  Section \ref{section4numerical} contains the numerical computation of the wave equation solutions, their associated eigenfunctions, and the Fourier approximations for the delta function.

\BREM{remHKE}
{Theoretically, the infinite propagation 
	speed for wave equation solutions was established in \cite{L}  on some p.c.f.
	fractals with heat kernel estimates  \begin{align}\label{hke}
	\frac{c_1}{V(x,t^{1/\beta})} \exp\left(-c_2 \left(\frac{d(x,y)^\beta}{t}\right)^{1/(\beta-1)}\right) &
	\leq p(t,x,y) \\ 
	\nonumber &\leq \frac{c_3}{V(x,t^{1/\beta})} \exp\left(-c_4 \left(\frac{d(x,y)^\beta}{t}\right)^{1/(\beta-1)}\right)
	\end{align}
for positive constants $c_1, c_2, c_3, c_4$, $x,y \in I$, $t\in (0,1]$, where $\beta=2/d_S$ and $V(x,r)=\mu(B_r(x))$.
 Kigami	in \cite{Ki04} obtained such estimates in a situation which 
 resembles, but is technically different, from  ours. 
 We conjecture that 
 an analogue \eqref{hke} holds in our 
 situation,   
 but proving this would lie outside of the scope of our paper. 
}

%
%
%

\subsubsection*{Acknowledgement} 
The authors are very grateful to 
Daniel Kelleher, Hugo Panzo and Antoni  Brzoska 
for many helpful discussions, and   to 
Luke Rogers for explaining the eigenfunction estimates based on his paper \cite{Ro}. 
\,A.T.~also thanks Sze-Man Ngai and 
Alexander Grigor'yan  
for very valuable advice. 
The authors thank anonymous referees 
for corrections and a substantial list of constructive 
suggestions leading to improvements in the first version of 
our paper, and for the 
suggestion to include the infinite wave propagation speed 
Remark~\ref{remHKE}.

\section{Eigenvalues of the fractal Laplacian on an interval}\label{section2eigenvalues}

\def\x{\sm{(x)}}
\def\ih#1{{{\huge\mbox{$#1$}}}}
\def\ih#1{$#1$}


In this section we define a 
particular self-similar structure on the unit interval.  In this way, it can be seen as a p.c.f fractal (see \cite{Ba,BNT,Ki1,Ki,T04,T07}).
In these papers the reader can find these definitions and an exposition of the general theory of Dirichlet forms on fractals, as well as further references on the subject.
%
%
Herein we will use three contractions for simplicity.  However, one could perform the same construction using any number of contractions in order to obtain a fractal Laplacian on the unit interval.

To define the standard Laplacian, we can use 
 three contractions 
$\,F_1,F_2,F_3:\mathbb R\to\mathbb R\,$ 
$
F_j(x)=\tfrac13x+\tfrac23p_j
$
with respective fixed points $p_1=0$, $p_2=\frac{1}{2}$, $p_3=1$. 
Then the interval $I{=}[0,1]$ 
is a unique compact 
 set such that 
$
I=
\bigcup\limits_{\sm{\sm{{j{=}1,2,3}}}}\hskip-.975ex
F_j(I).
$
The \bydef{discrete approximations} to $I$ are 
defined inductively by 
$
V_n = 
\hskip-.975ex\bigcup\limits_{\sm{\sm{{j{=}1,2,3}}}}\hskip-.975ex
F_j(V_{n\sm{\sm{-}}1})
=\big\{\tfrac k{3^n}\big\}_{k{=}0}^{3^n},
$
where $V_0=\partial I=\{0,1\}$ is the
\bydef{boundary} of $I$. For $x,y\in V_\ii {n}$ we write $y \sim x$ if $|x-y|=3^{-n}$. 
{Then the \bydef{standard discrete Dirichlet (energy) form} on $V_n$ is 
$$
\mathcal E_n (\ddD{f}) = 3^ n
{\sum\limits_ {\genfrac{}{}{0pt}{2}{x,y\in V_\ii {n}}{y \sim x}}}
(f\sm{(y)} {-} f\x )^\ii 2,
$$
and the \bydef{standard Dirichlet (energy) form} on $I$ is 
$
\mathcal E (\ddD{f}) = \lim\limits_\ii {n{\to}\infty} \mathcal E_n (\ddD{f}) 
$
if this limit exists.} 
%
%
We call a function $h$ harmonic if it minimizes the energy subject to the constraint of the given boundary values.
{Then we have that 
$
\mathcal E_{n{+}1} (\ddD{f})\geqslant\mathcal E_n (\ddD{f})
$ 
 for any function $f$, and 
$
\mathcal E_{n{+}1}(\ddD{h})= 
\mathcal E_n (\ddD{h})=\mathcal E(\ddD{h})
$
for a harmonic $h$. A function $h$ is harmonic if and only if it is linear. If $f$ is continuously differentiable 
then 
$$
\mathcal E (\ddD{f}) =\int_0^1|f'(x)|^2 dx.
$$
The domain $\mathcal F$ of this 
standard Dirichlet (energy) form $\mathcal E$ on $I$ 
coincides with the usual 
Sobolev space $H^1[0,1]$. 
Moreover $\mathcal E$ on $I$ is \bydef{self-similar} 
in the sense that 
$$
\mathcal E(\ddD{f})=3
\sum_{\sm{\sm{j=1,2,3}}}
\mathcal E(\ddD{f\sm{\circ}F_j}).
$$} 
{The corresponding \bydef{standard discrete \Lp s} on $V_n$ are 
$$
\Delta_n f\x = \tfrac12 
{\sum\limits_ {\genfrac{}{}{0pt}{2}{y\in V_\ii {n}}{y \sim x}}}
 (f(x) - f(y)) , \quad x \in V_n \backslash V_0,
$$
and the (renormalized) \Lp\ on $I$ is 
$$
\Delta f\x = \lim\limits_\ii {n{\to}\infty} 9^ n
\Delta_n f\x =-\frac12f''\x 
$$
for any twice differentiable function.}
In our convention the Laplacian is a nonnegative operator.
For any twice differentiable function $f$, the \bydef{Gauss--Green (integration by parts) formula applies}
$$
 \mathcal E (\ddD{f})=2\int_0^1f\Delta f dx + ff'\Big|^1_0. 
$$

We can   modify the above construction with the introduction of the parameter $p$, where $0 < p < 1$, and write $q = 1-p$. Later we will show that these parameters give the transition probabilities of a random walk on the unit interval. 
Now, we define contraction factors (or resistance weights) 
\begin{equation}
r_1=r_3=\frac{p}{1+p}\text{ \ \ and \ \ } 
r_2=\frac{q}{1+p}, 
\end{equation}
and 
measure weights
\begin{equation}
m_1=m_3=\frac{q}{1+q}\text{ \ \ and \ \ } 
m_2=\frac{p}{1+q}. 
\end{equation}
Note that in general the choices of resistance and measure weights are 
essentially free, up to constant multiples, according to Kigami's 
theory of 
 Harmonic calculus on p.c.f.\ self-similar sets \cite{Ki1,Ki}, 
 but we make a unique choice that leads to 
 a manageable spectral analysis, as explained in 
 \cite{eigen1,S,Sh,T,T04,T07}. We do not give a complete explanation here because it would require too much space. In short, the spectral decimation requires a symmetry 
 $m_1=m_3$. Moreover, 
 the spectral decimation also requires   that   the resistance weights are, up to a constant, reciprocals 
of the measure weights, and 
\begin{equation} m_1+m_2+m_3=r_1+r_2+r_3=1.\end{equation}
Thus, our system essentially has one independent parameter, which we denote $p$ and express everything else in terms of this parameter.

We may now define the three  contractions:
$\,F_1,F_2,F_3:\mathbb{R}\to\mathbb{R}$ 
with respective fixed points $p_1=0$, $p_2=\frac 12$, $p_3=1$ in terms of resistances which depend on our parameter $p$
\begin{equation}F_j(x)=r_jx+(1-r_j)p_j.\end{equation} 
Then the interval $I{=}[0,1]$ 
is the unique compact 
 set such that 
\begin{equation}
I=
\bigcup\limits_{\sm{\sm{{j{=}1,2,3}}}}\hskip-.975ex
F_j(I).\end{equation}
The \bydef{discrete approximations} to $I$ are 
defined inductively by 
\begin{equation}
V_n = 
\hskip-.975ex\bigcup\limits_{\sm{\sm{{j{=}1,2,3}}}}\hskip-.975ex
F_j(V_{n\sm{\sm{-}}1}),
\end{equation}
where $V_0=\partial I=\{0,1\}$ is the
\bydef{boundary} of $I$. 

The following definitions and results come directly 
from the more general theory in \cite{Ba,BNT,Ki1,Ki},
so we omit the proofs.

\BDF{denergyform}{The \bydef{discrete Dirichlet (energy) form} on $V_n$ is 
defined inductively 
\begin{equation}
\mathcal E_n (\ddD{f})=
\sum_{\sm{\sm{j=1,2,3}}}
\tfrac1{r_j} \mathcal E_{n-1}(f\sm{\circ}F_j).
\end{equation}
with $
\mathcal E_0 (\ddD{f}) = 
(f(1)-f(0))^2 $, 
and the \bydef{Dirichlet (energy) form} on $I$ is 
\begin{equation}
\mathcal E (\ddD{f}) = \lim\limits_\ii {n{\to}\infty} 
\mathcal E_n (\ddD{f}) =\int_0^1|f'(x)|^2 dx
\end{equation}
The domain $\mathcal F$ of $\E$ consists of continuous functions for which the limit is finite, 
and 
coincides with the usual 
Sobolev space $H^1[0,1]$. 
}

The existence of this limit is justified by the next proposition.


\BPROP{prop099}{We have that 
$
\mathcal E_{n{+}1} (\ddD{f})\geqslant\mathcal E_n (\ddD{f})
$ 
 for any function $f$, and 
\begin{equation}
\mathcal E_{n{+}1}(\ddD{h})= 
\mathcal E_n (\ddD{h})=\mathcal E(\ddD{h})
\end{equation}
for a harmonic function $h$.}

\BPROP{prop16}{
The Dirichlet (energy) form $\mathcal E$ on $I$ is local and regular, and is \bydef{self-similar} 
in the sense that 
\begin{equation}
\mathcal E(\ddD{f})=
\sum_{\sm{\sm{j=1,2,3}}}
\tfrac1{r_j} \mathcal E(\ddD{f\sm{\circ}F_j}).
\end{equation}
The domain of $\mathcal E$, see Definition~\ref{denergyform}, is dense in the space of continuous functions on $I$. 

The \bydef{$\mu$--\Lp} $\Delta_\mu$, satisfying  
the following 
Gauss--Green (integration by parts) formula 
\begin{equation}
 \mathcal E (\ddD{f})=
C\int_0^1f\Delta_\mu f d\mu + ff'\big|^1_0, 
\end{equation}
where $\mu$ is a unique probability \bydef{self-similar} 
measure 
with weights $m_1,$ $m_2,$ $m_3$, that is 
\begin{equation}
\mu =\sum_{\sm{\sm{j=1,2,3}}}m_j \mu \sm{\circ} F_j.
\end{equation}
can be defined by 
\begin{equation}\label{edL}
\Delta_\mu f\x 
= 
\hskip-.21em
\lim\limits_\ii {n{\to}\infty} 
\hskip-.21em
\big(1\sm{+}\tfrac2{pq}\big)^ n 
\Delta_n f\x ,
\end{equation}
where the discrete \Lp s 
\begin{equation}\label{Lp}
\Delta_n f(x_k)=\left\{
\begin{aligned}
\mbox{$f(x_k) - pf(x_{k-1}) - qf(x_{k+1}) $}&\\ 
& \ \ \text{or} \\
\mbox{$f(x_k) - qf(x_{k-1}) - pf(x_{k+1}) $}& \end{aligned}
\right. \ 
\end{equation}
 are defined as the 
generators of the nearest neighbor random walks on $V_n$ 
with transition 
probabilities $p$ and $q$ assigned according to the 
weights of the corresponding intervals. 
The domain of the corresponding continuous \Lp\ $\Delta_\mu $, 
defined  to be the set of all continuous function $f$ for which the limit 
\eqref{edL} exists and is continuous, 
is dense in the space of continuous functions on $I$. 
} 

Note that by definition 
$p=\frac{m_2}{m_1{+}m_2}$ and $q=\frac{m_1}{m_1{+}m_2}$. 
The transition 
probabilities $p$ and $q$ can be assigned inductively
as shown on \Fig{figRW}. 

\begin{proposition}[Self-similarity of the Laplacian]
\label{prop:Lapss}
\begin{equation}
\Delta_\mu(u\circ F_w) = \left(1+\frac{2}{pq}\right)^{-|w|}(\Delta_\mu u)\circ F_w.
\label{eq:Lapss}
\end{equation}
\end{proposition}

\BBF{figRW}
{Random walks corresponding to the discrete \Lp s $\Delta_n$.}
{\begin{center}
\begin{picture}(246,30)(0,-20) \setlength{\unitlength}{.45pt}
\thicklines
\put(0,0){\circle*{12}}
\put(540,0){\circle*{12}}
\put(0,0){\line(1,0){540}}
\put(3,-12){\vector(1,0){19}}
\put(537,-12){\vector(-1,0){19}}
\put(7,-33){{\scriptsize\ih1}}
\put(523,-33){{\scriptsize\ih1}}
\end{picture}

\begin{picture}(246,30)(0,-20) \setlength{\unitlength}{.45pt}
\thicklines
\put(0,0){\circle*{12}}
\multiput(0,0)(60,0){9}{\line(1,0){60}}
\multiput(180,0)(180,0){3}{\circle*{12}}
\multiput(3,-12)(180,0){3}{\vector(1,0){19}}
\multiput(177,-12)(180,0){3}{\vector(-1,0){19}}
\put(7,-33){{\scriptsize\ih1}}
\put(70,11){\ih{m_1}}
\put(250,11){\ih{m_2}}
\put(430,11){\ih{m_3}}
\put(163,-33){\ih q}
\put(187,-33){\ih p}
\put(343,-33){\ih p}
\put(367,-33){\ih q}
\put(523,-33){{\scriptsize\ih1}}
\end{picture}

\begin{picture}(246,30)(0,-20) \setlength{\unitlength}{.45pt}%
\put(0,0){\circle*{12}}
\multiput(0,0)(60,0){9}{\line(1,0){60}}
\multiput(60,0)(60,0){9}{\circle*{12}}
\multiput(3,-12)(60,0){9}{\vector(1,0){19}}
\multiput(57,-12)(60,0){9}{\vector(-1,0){19}}
\put(7,-33){{\scriptsize\ih1}}
\put(43,-33){\ih q}
\put(67,-33){\ih p}
\put(103,-33){\ih p}
\put(127,-33){\ih q}
\put(163,-33){\ih q}
\put(187,-33){\ih p}
\put(223,-33){\ih q}
\put(247,-33){\ih p}
\put(283,-33){\ih p}
\put(307,-33){\ih q}
\put(343,-33){\ih p}
\put(367,-33){\ih q}
\put(403,-33){\ih q}
\put(427,-33){\ih p}
\put(463,-33){\ih p}
\put(487,-33){\ih q}
\put(523,-33){{\scriptsize\ih1}}
\end{picture}
\end{center}
}


The above construction of the standard Laplacian and the associated Dirichlet form on $I$ corresponds to the case $p  = \frac{1}{2}$. In the $p \ne \frac{1}{2}$ case, a change of variables can either turn the Dirichlet form into the standard one, or turn the $\mu$-measure into Lebesgue measure, but \emph{not both} at the same time. For this reason, different values of $p$ give different $\mu$-Laplacians even up to a change of variable.

We can apply the classical result of Kigami and Lapidus \cite{KL} to show that  
both the Dirichlet and the Neumann \Lp s $\Delta_\mu$ 
satisfy the \bydef{spectral asymptotics}
\begin{equation}
0<
\liminf_{\lambda{\to}\infty}
\frac{\rho(\lambda)}{\lambda^{d_s/2}}
\leqslant
\limsup_{\lambda{\to}\infty}
\frac{\rho(\lambda)}{\lambda^{d_s/2}}
<\infty,
\end{equation}
where as before $\rho(\lambda)$ is the eigenvalue counting function, and the spectral dimension is 
\begin{equation}
d_s=\dfrac{\log9}{\log\big(1\sm{+}\tfrac2{pq}\big)}\leqslant1,
\end{equation}
where the inequality is strict if and only if $p\neq q$. 

{ 
In the lemma below,   $\sigma(\Delta_n)$ is the spectrum  n of the level $n$ Laplacian $\Delta_n$.}


\BLEM{prop17}{If $z\neq 1\pm p$,
then $ R(z)\in\sigma(\Delta_\ii {n})$ 
if and only if 
$z\in\sigma(\Delta_\ii {n\ps+1})$, with the same multiplicities. 
Here 
\begin{equation} R(z)\ps=\frac{z( z^2 \pms-3z\pms+2\pms+pq)}{pq}.\end{equation}
Moreover,  the Neumann discrete 
\Lp s have  simple spectrum with 
$\sigma(\Delta_\ii {0})=\{0,2\}$ 
and
\begin{equation}
\sigma(\Delta_\ii {n})=\{0,2\}\bigcup_{m=0}^{n-1}R^{-m}\{1\pm q\}
\end{equation}
for 
all $n>0$. 
In particular, for all $n>0$ we have  $0,1\pm q,2\in\sigma(\Delta_\ii {n})$. 
Also, for 
all $n>0$ we have  $1\pm p\in\sigma(\Delta_\ii {n})$ 
if and only if $p=q$.} 

\Pr 
In this case, according to \cite[Lemma 3.4]{T}, \cite[(3.2)]{MT}, we have that $R(z)=\dfrac{\varphi_1(z)}{\varphi_0(z)}$, where 
$\varphi_0$ and $\varphi_1$ solve the matrix equation 
\begin{equation}
S - z I_0 -  \bar   X    (Q-z I_1)^{-1}   X    =\varphi_0 (z) H_0  
-\varphi_1 (z) I_0.
\end{equation}
with $S=I_0=I_1=I_{2\times 2}, \; X=-qI_{2\times 2}, \; \bar X = -I_{2\times 2}$,
\begin{equation}
Q=\left(\begin{array}{rr}
 1 &  -p\cr
 -p & 1\cr
\end{array} \right),
\end{equation}
and
\begin{equation}
H_0=\left(\begin{array}{rr}
 1 &  -1\cr
 -1 & 1\cr
\end{array} \right).
\end{equation}
Solving this we obtain 
\begin{equation}
\varphi_0(z)=\frac{pq}{z^2-2z+1-p^2}
\end{equation}
and 
\begin{equation}
\varphi_1(z)=\frac{z( z^2 \pms-3z\pms+2\pms+pq)}{z^2-2z+1-p^2}.
\end{equation}
 Then we use the abstract spectral self-similarity results (see \cite{T,MT}) to find that $\sigma(\Delta_\ii {n+1}) = R^{-1} \{\sigma(\Delta_\ii {n})\}$.
Note that $0$ and $2$ are fixed points of $R(z)$. 
The preimages of $0$ are $0$, $1+p$ and $1+q$. 
The preimages of  $2$ are $2$, $1-p$ and $1-q$. If $p\neq q$ then $1\pm p$ are not 
eigenvalues because they are poles of $\varphi_0(z)$ (see \cite{T,MT}). 
\rP

\BBF{figCD}
{Sketch of the cubic polynomial $R(z)$ 
associated with the fractal \Lp s on the interval.}
{\begin{picture}(120,133)(-50,-66) \setlength{\unitlength}{0.5pt}\small
\thicklines
\put(-120,-100){\vector(1,0){240}}
\put(-100,-120){\vector(0,1){240}}
\put(-80,-125){$\ii \max(p,q)$}
\put(-155,-125){$\ii{(0, 0)}$}
\put(110,107){$\ii{(2, 2)}$}
\put(100,100){\circle*{5}}
\put(-20,100){\circle*{5}}
\put(-20,-100){\circle*{5}}
\put(20,-100){\circle*{5}}
\put(100,-100){\circle*{5}}
\put(-100,-100){\circle*{5}}
\qbezier[35](100,-100)(100,0)(100,100)
\qbezier[35](-20,-100)(-20,0)(-20,100)
\qbezier[35](-100,-100)(0,-100)(100,-100)
\qbezier(-100,-100)(-45,307)(0,0)
\qbezier[63](-100,-100)(-45,241.5)(0,0)
\qbezier[35](100,100)(0,100)(-100,100)
\qbezier(100,100)(45,-307)(0,0)
\qbezier[63](100,100)(45,-241.5)(0,0)
\end{picture}}

\BREM{remFI}
{In \Fig{figCD} we give a sketch that describes 
the complex dynamics of the family of cubic polynomials 
associated with the fractal \Lp s on the interval (see  \cite{T,MT}). 
The curved dotted line corresponds to the case when $p=\frac12$ and 
the Julia set is the interval $[0, 2]$. For any other value of $p$ ($0<p<1$, $p\neq \frac{1}{2}$),
the graph of the polynomial $R(z)$ behaves like the shown solid curved line.
It is easy to see that then the 
Julia set of $R(z)$ is a Cantor set of Lebesgue measure zero. 
Note that the transformation $p\mapsto1-p$ does not change the polynomial 
$R(z)$, although the \Lp s ${\Delta_\mu}$ are different.} 


\begin{figure}
\begin{center}
\includegraphics[width=0.95\textwidth]{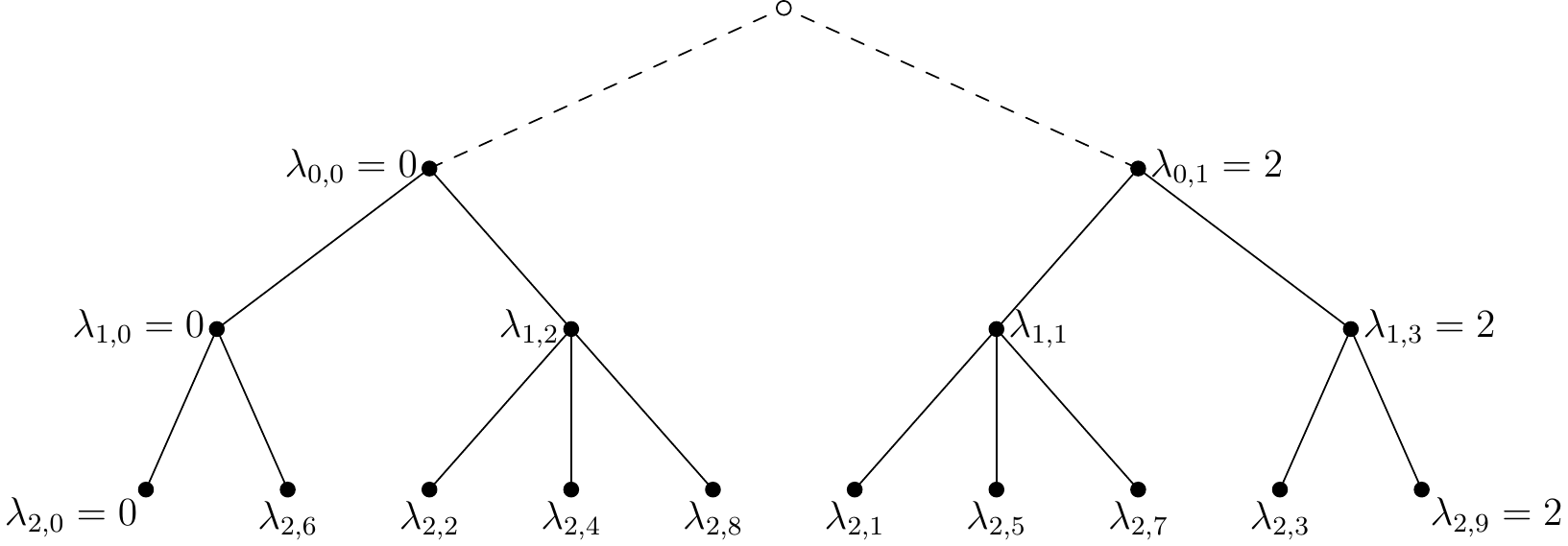}
\caption{Eigenvalues of the first three fractal levels arranged in a rooted tree
and numbered in increasing order.}
\label{fig:tree}
\end{center}
\end{figure}

\section{Spectral Decimation and Eigenfunction Approximations in the limit}\label{section3eigenfunctions}
 
 \begin{figure}[!]
 \includegraphics[width=\textwidth]{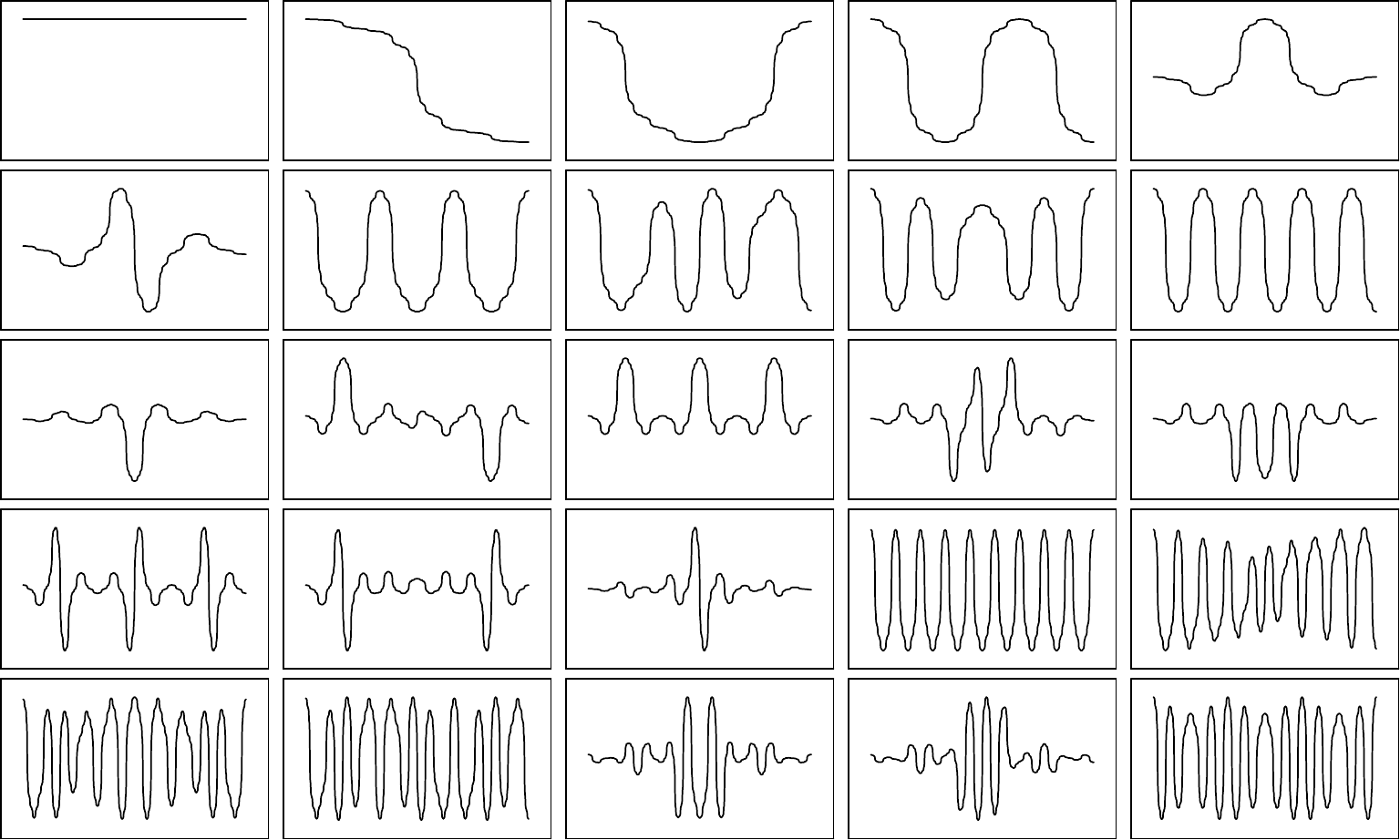}
 \caption{The first 25 eigenfunctions of the fractal Laplacian with $p = \frac{1}{5}$}\label{f1}
 \end{figure}
 
 \begin{figure}[!]
 \includegraphics[width=\textwidth]{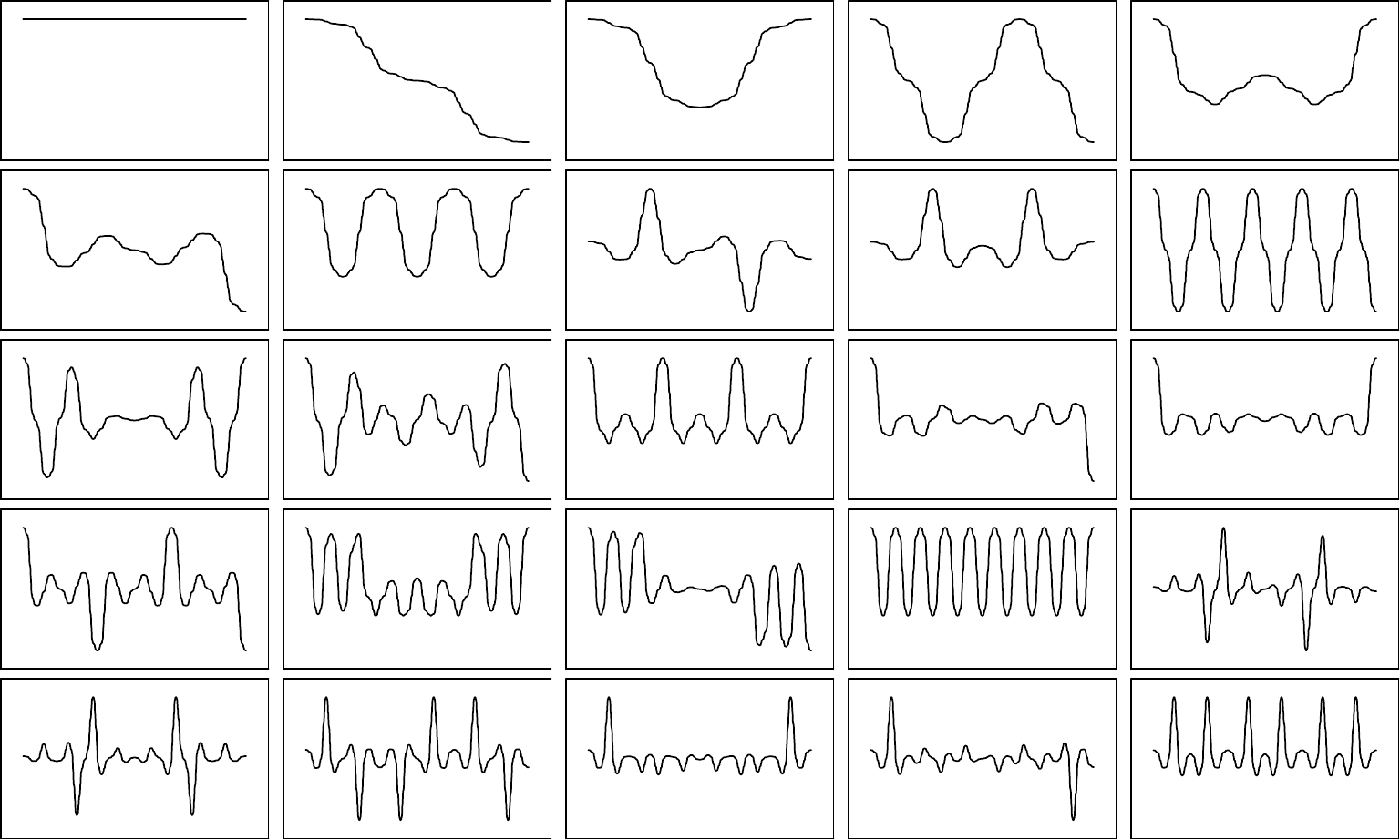}
 \caption{The first 25 eigenfunctions of the fractal Laplacian with $p = \frac{4}{5}$}\label{f2}
 \end{figure}
 

 \begin{figure}[!]
  \begin{center}
 \includegraphics[scale=0.59]{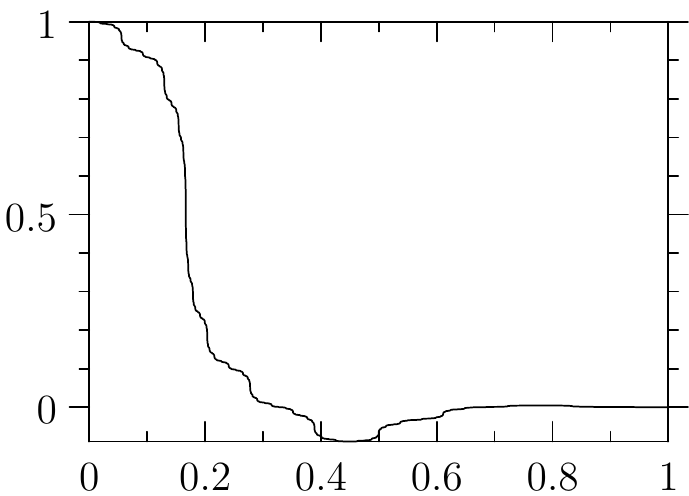}
 \includegraphics[scale=0.59]{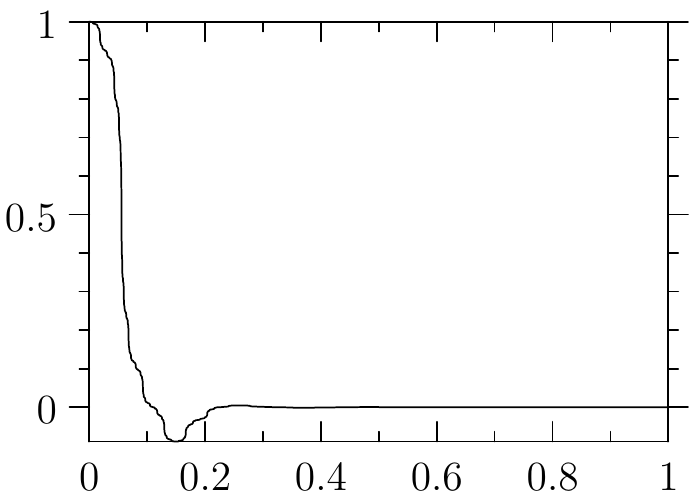}
 \includegraphics[scale=0.59]{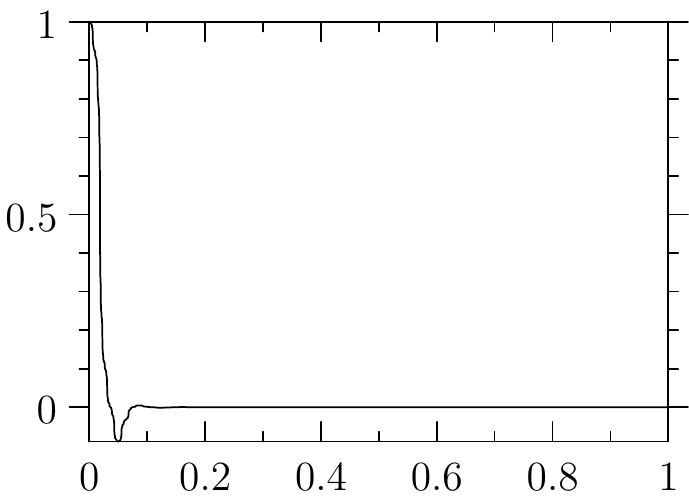}
 \end{center}
 \caption{Fourier approximations for delta function, $p = \frac{1}{5}$.
      Left to right: $n_0 = 2, 3, 4$.}\label{fdelta1}
 \end{figure}

 \begin{figure}[!]
 \begin{center}
 \includegraphics[scale=0.59]{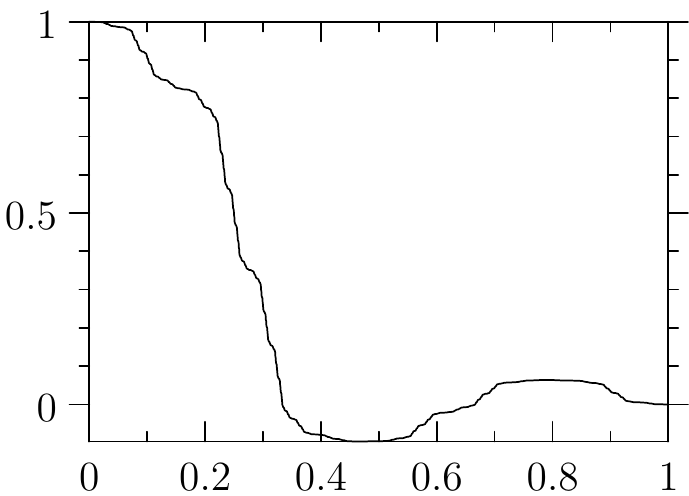}
 \includegraphics[scale=0.59]{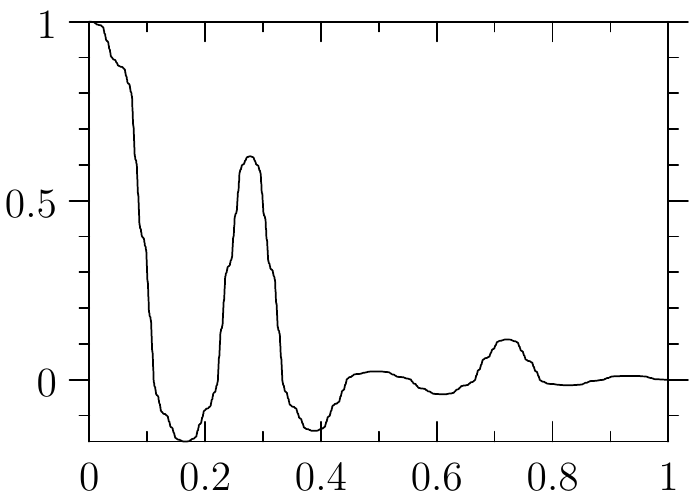}
 \includegraphics[scale=0.59]{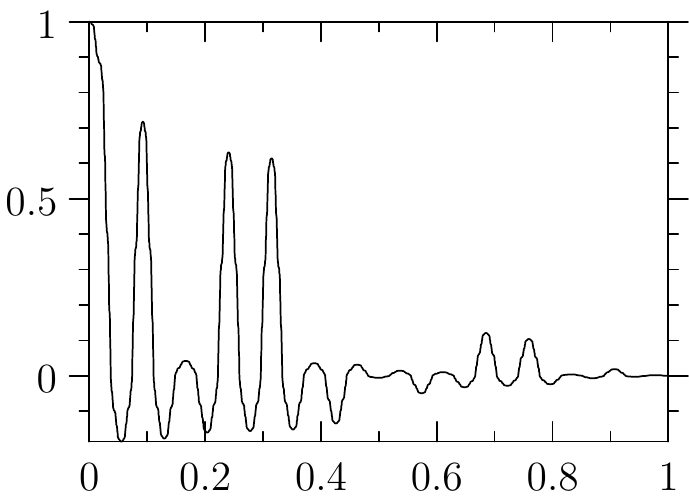}
 \end{center}
 \caption{Fourier approximations for delta function, $p = \frac{4}{5}$.
      Left to right: $n_0 = 2, 3, 4$.}
 \label{fdelta2}
 \end{figure}

Thus far we have described the spectral decimation which allows us to characterize the eigenvalues of the fractal Laplacian (Lemma \ref{prop17}). We now turn to the eigenfunctions.

\subsection{Eigenfunction extension}

In this subsection we demonstrate how to extend an eigenfunction $f_{n,*}$ to an eigenfunction $f_{n+1,*}$ using spectral decimation.

To fix notation, let $x_0<y_0<y_1<x_1$ be four consecutive vertices in $V_{n+1}$ with $x_0,x_1 \in V_n$ and $y_0, y_1\in V_{n+1}\setminus V_n$. Given an eigenfunction $f_{n,*}$ of $\Delta_n$ with eigenvalue $\lambda_{n,*}$, we define its extension $f_{n+1,*}$ to $V_{n+1}$ according to the formulas
\begin{equation}\label{ext0}
f_{n+1,*}(y_0)  =  \frac{q(1-z)f_{n,*}(x_0) + pqf_{n,*}(x_1)}{(1-p-z)(1+p-z)},
\end{equation}
\begin{equation}\label{ext1}
f_{n+1,*}(y_1)  =  \frac{q(1-z)f_{n,*}(x_1) + pqf_{n,*}(x_0)}{(1-p-z)(1+p-z)}.
\end{equation}
Here we assume $z \neq 1\pm p$. The claim is that $f_{n+1,*}$ is an eigenfunction of $\Delta_{n+1}$ with eigenvalue $z=R^{-1}(\lambda_{n,*})$, where $R$ is the cubic polynomial which appeared in Lemma \ref{prop17}. As explained in the proof of \ref{prop17}, the preimage $R^{-1}([0,2])$ has three branches, so each eigenvalue $\lambda_{n,*}$ on level $n$ generates three new eigenvalues $\lambda_{n+1,*}$ on level $(n+1)$. The only exceptions are the eigenvalues $0$ and $2$, each of which generates two new eigenvalues because $1\pm p$ are forbidden (see Figure \ref{fig:tree}). This means that each eigenfunction extends to either two or three eigenfunctions at the next level.

\begin{theorem}[Eigenfunction extension]
\label{thm:eigext}
Suppose $f_{n,*}: V_n\to\mathbb{R}$ is an eigenfunction of $\Delta_n$ with eigenvalue $\lambda_{n,*}$. Let $f_{n+1,*}: V_{n+1}\to\mathbb{R}$ be an extension of $f_{n,*}$ to $V_{n+1}$ defined via (\ref{ext0}) and (\ref{ext1}), with $z \neq 1\pm p$. If 
\begin{equation}\label{EVcubic}
\lambda_{n,*} = R(z) = \frac{z^3 - 3z^2 + (2+pq)z}{pq}, 
\end{equation}
then $f_{n+1,*}$ is an eigenfunction of $\Delta_{n+1}$ with eigenvalue $z = R^{-1}(\lambda_{n,*})$.
\end{theorem}

\begin{proof}
We break the proof into two parts. Given $f_{n,*}$, we first show that the following are equivalent for an extension $f_{n+1,*}$ of $f_{n,*}$:
\begin{enumerate}
\item $f_{n+1,*}$ is defined via the extension formulas (\ref{ext0}) and (\ref{ext1}).
\item $f_{n+1,*}$ satisfies the eigenvalue equation $\Delta_{n+1}f_{n+1,*}=z f_{n+1,*}$ on $V_{n+1}\setminus V_n$.
\end{enumerate}
After establishing this equivalence, we proceed to show that $f_{n+1,*}$ is an eigenfunction of $\Delta_{n+1}$ on all of $V_{n+1}$, provided that (\ref{EVcubic}) holds.

First we show the equivalence of (1) and (2). Assuming (2), we apply the eigenvalue equation $\Delta_{n+1}f_{n+1,*} = zf_{n+1,*}$ at the points $y_0,y_1 \in V_{n+1}\setminus V_n$ to obtain, by using both formulae in 
(\ref{Lp}) depending on the point $x$ (in fact, to cover these cases as well as the case $x\in V_0$, i.e. $x$ is a boundary vertex, below we   use the parameters $a\in\{p,q,0,1\}$, and $b=1-a$, instead of $p$ and $q=1-p$), 
\begin{eqnarray}
 \label{twoeig}    (1-z)f_{n+1,*}(y_0) &=& pf_{n+1,*}(y_1) +qf_{n+1,*}(x_0), \\
    (1-z)f_{n+1,*}(y_1) &=& qf_{n+1,*}(x_1) +pf_{n+1,*}(y_0). 
\end{eqnarray}
This is a linear system of 2 equations with 2 unknowns ($f_{n+1,*}(x_0)$ and $f_{n+1,*}(x_1)$ are known, $f_{n+1,*}(y_0)$ and $f_{n+1,*}(y_1)$ are unknown), which has a unique solution. After some elementary calculation, and using the fact that $f_{n+1,*}|_{V_n} = f_{n,*}$, it is easy to verify that $f_{n+1,*}(y_0)$ and $f_{n+1,*}(y_1)$ are uniquely expressed in terms of $f_{n,*}(x_0)$ and $f_{n,*}(x_1)$ according to the extension formulas (\ref{ext0}) and (\ref{ext1}), which shows (1). The reverse implication (1) $\Rightarrow$ (2) is straightforward.

At this point we have proved that the eigenvalue equation $\Delta_{n+1} f_{n+1,*}(x) = zf_{n+1,*}(x)$ holds for $x\in V_{n+1}\setminus V_n$. However, we have neither used the property of the eigenfunction $f_{n,*}$, nor related $z$ to the eigenvalue $\lambda_{n,*}$. To do so we must check the $\Delta_{n+1}$-eigenvalue equation on $V_n$. 

We introduce some additional notation. Fix an $x\in V_n$. Let $x'_0, x'_1 \in V_n$ be adjacent to $x$ on level $n$, and $y'_0, y'_1 \in V_{n+1}\setminus V_n$ be adjacent to $x$ on level $(n+1)$, as shown in Figure \ref{celldiagram}. (If $x\in V_0 =\{0,1\}$, then there is only one adjacent vertex on level $n$. This will be taken care of in the next argument.) We also label the transition probabilities according to (\ref{Lp}); see also Figure \ref{figRW}. The parameter $a$ can be one of $\{p,q,0,1\}$ depending on $x$. In particular, to take into account that $x\in V_0$ has only 1 adjacent vertex, we set $a=0$ if $x=1$ and $a=1$ if $x=0$. The parameter $b$ is set to equal $1-a$.

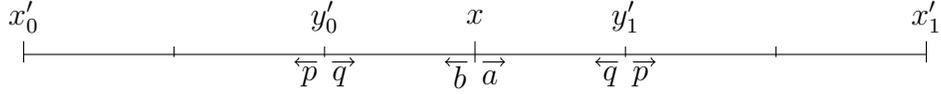
\begin{figure}
\begin{center}
\begin{tikzpicture}
\draw (1,0) -- (13,0);
\foreach \x  in {1,7,13}
\draw[xshift=\x cm] (0pt,5pt) -- (0pt,-3pt); 
\foreach \x  in {3,5,9,11}
\draw[xshift=\x cm] (0pt,3pt) -- (0pt,-1pt);
\node at (1, 0.5) {$x'_0$};
\node at (7, 0.5) {$x$};
\node at (13, 0.5) {$x'_1$};
\node at (5, 0.5) {$y'_0$};
\node at (9, 0.5) {$y'_1$};
\node at (7.2,-.3){$a$};
\node at (6.8,-.3){$b$};
\node at (5.2,-.3){$q$};
\node at (4.8,-.3){$p$};
\node at (9.2,-.3){$p$};
\node at (8.8,-.3){$q$};
\draw[->](7.1,-0.1) -- (7.4,-0.1);
\draw[<-](6.6,-0.1) -- (6.9,-0.1);
\draw[->](5.1,-0.1) -- (5.4,-0.1);
\draw[<-](4.6,-0.1) -- (4.9,-0.1);
\draw[->](9.1,-0.1) -- (9.4,-0.1);
\draw[<-](8.6,-0.1) -- (8.9,-0.1);
\end{tikzpicture}
\end{center}
\caption{A diagram of two adjacent $(n+1)$-cells used in the proof of Theorem \ref{thm:eigext}.}
\label{celldiagram}
\end{figure}

Now we show that if $z\neq 1\pm p$ and (\ref{EVcubic}) holds, then $\Delta_{n+1}f_{n+1,*}(x) = zf_{n+1,*}(x)$ for $x\in V_n$. By (\ref{Lp}),
\begin{eqnarray}
\label{eq:n+1} \Delta_{n+1} f_{n+1,*}(x) &=& f_{n+1,*}(x) - a f_{n+1,*}(y'_1) - b f_{n+1,*}(y'_0),\\
\label{eq:n} \Delta_n f_{n,*}(x) &=& f_{n,*}(x) - a f_{n,*}(x'_1) - b f_{n,*}(x'_0).
\end{eqnarray}
Using the extension formulas (\ref{ext0}), $(\ref{ext1})$, $f_{n+1,*}(x)=f_{n,*}(x)$, and (\ref{eq:n+1}), we find
\begin{eqnarray}
\nonumber &&(\Delta_{n+1}-z)f_{n+1,*}(x) \\
\nonumber &=& (1-z) f_{n,*}(x) - a\left(\frac{q(1-z)f_{n,*}(x) + pqf_{n,*}(x'_1)}{(1-p-z)(1+p-z)} \right) \\
\nonumber && \qquad \qquad \qquad - ~b\left( \frac{q(1-z)f_{n,*}(x) + pqf_{n,*}(x'_0)}{(1-p-z)(1+p-z)}\right) \\
\nonumber &=& \frac{(1-z)(1-p-z)(1+p-z)-q(1-z)}{(1-p-z)(1+p-z)} f_{n,*}(x)\\
&& \label{eq:long} \quad -~\frac{pq\left(a f_{n,*}(x'_1) + b f_{n,*}(x'_0)\right)}{(1-p-z)(1+p-z)}.
\end{eqnarray}
Using (\ref{eq:n}) we can write 
\begin{equation}
af_{n,*}(x'_1) + bf_{n,*}(x'_0) = -(\Delta_n-1)f_{n,*}(x) = - (\lambda_{n,*}-1) f_{n,*}(x).
\end{equation}
This allows us to replace the second term of (\ref{eq:long}), so that the entire (\ref{eq:long}) equals
\begin{eqnarray}
\nonumber && \frac{(1-z)[(1-p-z)(1+p-z)-q]-pq (1-\lambda_{n,*})}{(1-p-z)(1+p-z)} f_{n,*}(x) \\
&=& \frac{-(z^3 -3z^2 +(2+pq)z )+pq -pq(1-\lambda_{n,*})}{(1-p-z)(1+p-z)} f_{n,*}(x).
\end{eqnarray}
Infer that $(\Delta_{n+1}-z)f_{n+1,*}=0$ on $V_n$, and in turn on $V_{n+1}$, if $z\neq 1\pm p$ and
\begin{equation}
\lambda_{n,*} = \frac{z^3-3z^2+(2+pq)z}{pq} = R(z).
\end{equation}
\end{proof}

\subsection{Continuity in the limit}

In this subsection, we show that the eigenfunction extension algorithm (Theorem \ref{thm:eigext}) produces a continuous eigenfunction of the fractal Laplacian in the limit $n\to\infty$, provided that one always chooses the lowest branch of the inverse map $R^{-1}$ at all levels $n \geq n_0$.

\begin{lemma}\label{bddlemma}
Fix $n_0,k \in \mathbb{N} \cup \{0\}$. Let $f_{n_0,k}:V_{n_0} \to \mathbb{R}$ be an eigenfunction of $\Delta_{n_0}$ with eigenvalue $\lambda_{n_0,k}$. Let $\{f_{n_0+i,k}\}_{i=1}^\infty$ be the sequence of $\Delta_{n_0+i}$-eigenfunctions extended from $f_{n_0,k}$ via successive applications of Theorem \ref{thm:eigext}, where one always chooses the lowest branch of the inverse cubic polynomial $R^{-1}(z)$ [see (\ref{EVcubic})]. Then 
\begin{equation}
\limsup_{i\to\infty} \max_{x\in V_{n_0+i}} |f_{n_0+i,k}(x)|
\end{equation}
 is bounded.
\end{lemma}
\begin{proof}
From Lemma \ref{prop17} we know that $\lambda_{n,*} = R(\lambda_{n+1,*})$. Assume that the lowest branch of $R^{-1}$ is chosen to generate $\lambda_{n+1,k} = R^{-1}(\lambda_{n,k})$ from $\lambda_{n,k}$. Observe that $R$ is concave on $[0,\min(p,q)]$; therefore the graph of $R$ on $[0,\min(p,q)]$ lies above the secant line connecting $(0,0)$ and $(\min(p,q),2)$ (see Figure \ref{figCD}). This implies the inequality
\begin{equation}
\lambda_{n,k} = R(\lambda_{n+1,k})\geq \frac{2}{\min(p,q)}\lambda_{n+1,k}.
\end{equation}
By iterating this inequality, we see that the $i$-fold preimage $\lambda_{n_0+i,k}=R^{-i}(\lambda_{n_0,k})$, where the lowest branch of $R^{-1}$ is always chosen, satisfies
\begin{equation}\label{evineq}
\lambda_{n_0+i,k} \le \left(\frac{\min(p,q)}{2}\right)^i \lambda_{n_0,k}.
\end{equation}
The corresponding eigenfunctions $f_{n_0+i,k}$ are generated via Theorem \ref{thm:eigext}.

Let $M_{n,k}=\max_{x\in V_n}|f_{n,k}(x)|$. For each $n> n_0$ and each $y\in V_{n+1}\setminus V_n$, we use the eigenfunction extension algorithm (\ref{ext0}) and (\ref{ext1}) to arrive at the following estimate: there exist $x_0, x_1 \in V_n$ such that
\begin{eqnarray}
|f_{n+1,k}(y)| &=&\left|\frac{q(1-\lambda_{n+1,k})f_{n+1,k}(x_0) + pqf_{n+1,k}(x_1)}{(1-p-\lambda_{n+1,k})(1+p-\lambda_{n+1,k})}\right| \\
&\leq& \frac{q(1-\lambda_{n+1,k})|f_{n+1,k}(x_0)| + pq|f_{n+1,k}(x_1)|}{(1-p-\lambda_{n+1,k})(1+p-\lambda_{n+1,k})} \\
&\leq& \frac{q}{q-\lambda_{n+1,k}} M_{n,k}.
\end{eqnarray} 
In the second line we used the triangle inequality and the bound $\lambda_{n+1,k} <1$, which can be seen from (\ref{evineq}). This then implies the estimate
\begin{equation}\label{maxIEQ}
M_{n+1,k} \le \frac{q}{q-\lambda_{n+1,k}}M_{n,k}.
\end{equation}
for all $n\geq n_0$. Applying (\ref{maxIEQ}) inductively and using (\ref{evineq}), we see that for all $i\in \mathbb{N}$,
\begin{equation}
M_{n_0+i,k} \le M_{n_0,k} \prod_{j=1}^i \frac{q}{q-\lambda_{n_0+j,k}} \leq M_{n_0,k}\prod_{j=1}^i \left(1-\left(\frac{\min(p,q)}{2}\right)^j \frac{\lambda_{n_0,k}}{q}\right)^{-1}.
\end{equation}
Setting $\gamma_j = \left(\frac{\min(p,q)}{2}\right)^j \frac{\lambda_{n_0,k}}{q}$ and taking the limit, we obtain
\begin{equation}
\limsup_{i\to\infty} M_{n_0+i,k} \leq M_{n_0,k} \limsup_{i\to\infty} \prod_{j=1}^i \left(1-\gamma_j\right)^{-1}.
\end{equation}

It remains to verify the convergence of the infinite product $\prod_{j=1}^\infty (1-\gamma_j)^{-1}$, which is equivalent to showing the convergence of the series $\sum_{j=1}^\infty \log(1-\gamma_j)^{-1}$. Observe that if we set $\epsilon_j$ to satisfy $(1-\gamma_j)^{-1}=1+\epsilon_j$, then
\begin{equation}
\log(1-\gamma_j)^{-1} = \log(1+\epsilon_j) \leq \epsilon_j 
\end{equation}
by the inequality $1+x\leq e^x$. Moreover, since $\gamma_j = K r^j$ for suitable positive constants $K$ and $r\leq \frac{1}{2}$, we can always find a constant $K_0$ such that
\begin{equation}
\epsilon_j = \frac{\gamma_j}{1-\gamma_j} = \frac{K r^j}{1-K r^j} \leq K_0 r^j
\end{equation}
for all sufficiently large $j$. Since the geometric series $\sum_j K_0 r^j$ converges, this implies that the series $\sum_{j=1}^\infty \log(1-\alpha_j)^{-1}$ converges.
\end{proof}

Now we prove the continuity of the eigenfunction in the limit.

\begin{theorem}
\label{thm:eigfconti}
Let $\{f_{n_0+i,k}\}_{i=1}^\infty$ be the sequence of $\Delta_{n_0+i}$-eigenfunctions extended from $f_{n_0,k}$ as in Lemma \ref{bddlemma}. Then the limit $f_k := \lim_{i\to\infty} f_{n_0+i,k}$ is uniform on $I$, and can be extended to a continuous function on $I$.
\end{theorem}

\begin{proof}
The key argument is that since the eigenvalues $\lambda_{n_0+i,k} = R^{-i}(\lambda_{n_0,k})$ tend to $0$ as $i\to\infty$, the eigenfunction extension (\ref{ext0}) and (\ref{ext1}) of $f_{n_0+i,k}$ to $V_{n_0+i+1}$ can be approximated by the harmonic extension of $f_{n_0+i,k}$ to $V_{n_0+i+1}$ as $i\to\infty$, uniformly on $I$. Since a harmonic extension on $I$ is continuous in the limit, we  deduce that the limit $f_k$ can also be extended to a continuous function.

Define, for each $n\geq n_0$ and each $k$, the harmonic extension $\tilde{f}_{n+1,k}$ of $f_{n,k}$ to $V_{n+1}$. Using the coordinates $x_0$, $x_1$, $y_0$, $y_1$ introduced before Theorem \ref{thm:eigext},
\begin{eqnarray}
\tilde{f}_{n+1,k}(y_0)  &=& \frac{f_{n,k}(x_0) + p f_{n,k}(x_1)}{1+p}, \\
\tilde{f}_{n+1,k}(y_1)  &=& \frac{f_{n,k}(x_1) + p f_{n,k}(x_0)}{1+p}.
\end{eqnarray}
 Note that these are (\ref{ext0}) and (\ref{ext1}) with $z=0$. 

%

Let us now estimate $|f_{n+1,k}(y_0)-\tilde{f}_{n+1,k}(y_0)| $, which equals
\begin{eqnarray}\label{showlin}
&&\left|\frac{q(1-\lambda_{n,k})f_{n,k}(x_0)+pq f_{n,k}(x_1)}{(1-p-\lambda_{n,k})(1+p-\lambda_{n,k})}-\frac{f_{n,k}(x_0)+p f_{n,k}(x_1) }{1+p}\right|\\
&=& \left|\frac{\lambda_{n,k}(1+p-pq-\lambda_{n,k}) f_{n,k}(x_0) +p\lambda_{n,k}(2-\lambda_{n,k}) f_{n,k}(x_1)}{(q-\lambda_{n,k})(1+p-\lambda_{n,k})(1+p)} \right| \label{fdiff}
\end{eqnarray}
Using the triangle inequality and then replacing $|f_{n,k}(x_0)|$ and $|f_{n,k}(x_1)|$ by the sup $M_{n,k} = \sup_{x\in V_n}|f_{n,k}(x)|$, we can bound (\ref{fdiff}) from above by
\begin{equation}
\frac{\lambda_{n,k} (1+p-pq-\lambda_{n,k}) M_{n,k} + p\lambda_{n,k}(2-\lambda_{n,k}) M_{n,k}}{(q-\lambda_{n,k})(1+p-\lambda_{n,k})(1+p)}\\
= \frac{\lambda_{n,k}}{q-\lambda_{n,k}} M_{n,k}
\end{equation}
Since $\lim_{n\to\infty} M_{n,k}$ is bounded by Lemma \ref{bddlemma} and $\lambda_{n,k} \to 0$ as $n\to\infty$, the right-hand side of this inequality tends to 0. The same estimate holds for $|f_{n+1,k}(y_1)-\tilde{f}_{n+1,k}(y_1)| $. Since $y_0$ and $y_1$ are arbitrary, we conclude that $|f_{n,k} - \tilde{f}_{n,k}|$ converges to $0$ uniformly on $I$.
\end{proof}

\subsection{Spectral decomposition of the delta function} \label{sec:specdecom}

Let $\{f_{n,k}\}_{k=0}^{3^n}$ be a complete set of eigenfunctions of $\Delta_n$ with corresponding eigenvalues $\{\lambda_{n,k}\}_{k=0}^{3^n}$. Consider the level-$n$ delta function $\delta_0^{(n)}: V_n \to\mathbb{R}$ defined by
\begin{equation}
\delta^{(n)}_0(x) = \left\{\begin{array}{ll} 1, & \text{if}~x=0,\\ 0, & \text{if}~x \in V_n \backslash \{0\}.\end{array}\right.
\end{equation}
By the spectral theorem, we can find a set of real numbers (or weights) $\{\alpha_{n,k}\}_{k=0}^{3^n}$ such that 
\begin{equation}
\delta_0^{(n)}(x) = \sum_{k=0}^{3^n}\alpha_{n,k}f_{n,k}(x).
\end{equation}
The sequence $\delta_0^{(n)}$ approximates a delta function at $0$ in the limit $n\to\infty$.

In order to study the wave equation in Section \ref{sectionwaveequation}, we need estimates on the eigensolutions $(f_{n,*}, \lambda_{n,*})$, as well as information about the weights $\alpha_{n,*}$. We will address the former in Section \ref{section5estimates}, and the latter in the following proposition.

\begin{proposition}
The weights $\alpha_{n+1,*}$ can be obtained inductively from $\alpha_{n,*}$.
\label{prop:weights}
\end{proposition}

\begin{proof}
First we fix our convention at $n=0$. The two (non-$\ell^2$-normalized) eigenfunctions $f_{0,1}$ and $f_{0,2}$ of $\Delta_0$ are 
\begin{equation}
(f_{0,1}(0), f_{0,1}(1))=(1,1) ~\text{and}~ (f_{0,2}(0), f_{0,2}(1))=(1,-1),
\end{equation}
with corresponding eigenvalue $0$ and $2$, respectively. It is then easy to see that $\delta_0^{(0)} = \frac{1}{2} f_{0,1} + \frac{1}{2} f_{0,2}$, \emph{i.e.,} $\alpha_{0,1}=\alpha_{0,2}=\frac{1}{2}$. 

For the iteration step, suppose the weights $\alpha_{n,*}$ are known at level $n$, and we want to determine the weights $\alpha_{n+1,*}$. The idea is to write each contribution $\alpha_{n,k}f_{n,k}$ in terms of a linear combination $\sum_j \alpha_{n+1,k_j} f_{n+1,k_j}$ of the (2 or 3) eigenfunctions $f_{n+1,k_j}$ which are extensions of $f_{n,k}$ given by Theorem \ref{thm:eigext}.

To make this idea precise without adding too much notation, we fix $n$ and $k$, and write $f$, $\lambda$, and $\alpha$ as respective shorthands for $f_{n,k}$, $\lambda_{n,k}$, and $\alpha_{n,k}$. If $\lambda \notin \{0,2\}$, then spectral decimation (Theorem \ref{thm:eigext}) implies that $f$ has 3 extensions $f_1$, $f_2$, and $f_3$ to $V_{n+1}$ which are eigenfunctions of $\Delta_{n+1}$ with respective eigenvalues $\lambda_1 \leq \lambda_2 \leq \lambda_3$. We would like to find the corresponding weights $\alpha_1$, $\alpha_2$, and $\alpha_3$ by imposing the following matching condition: For any four consecutive vertices $x_0< y_0<y_1< x_1$ in $V_{n+1}$ with $x_0, x_1 \in V_n$ and $y_0, y_1 \in V_{n+1}\setminus V_n$,
%
\begin{equation}\label{wsplit}
\sum_{i=1}^3 \alpha_i f_i(x) =
\begin{cases}
 \alpha f(x_0), & \text{if}~x=x_0,\\
\alpha f(x_1), & \text{if}~x\in \{x_1,y_0,y_1\}.
\end{cases} 
\end{equation}
An explicit calculation verifies that with the weights $\alpha_{n+1,*}$ generated from this matching condition, we have
\begin{equation}
\sum_{k=0}^{3^{n+1}} \alpha_{n+1,k}f_{n+1,k} =  \delta_0^{(n+1)}.
\end{equation}

We now determine the weights. Observe that $f_i$ and $f$ agree on $V_n$ by construction. This together with the matching condition (\ref{wsplit}) at $x_0$ (or at $x_1$) implies that
\begin{equation}\label{specdecomp1}
\sum_{i=1}^3 \alpha_i=\alpha.
\end{equation}
Next, using the eigenfunction extension formula (\ref{ext0}) and the matching condition at $y_0$ in (\ref{wsplit}), we get
\begin{equation}\label{leftsum}
\sum_{i=1}^3\frac{q(1-\lambda_i)}{(1-p-\lambda_i)(1+p-\lambda_i)}\alpha_i = 0.
\end{equation}
Notice that there is no dependence on $f$. Similarly, using (\ref{ext1}) and the matching condition at $y_1$ in (\ref{wsplit}), we arrive at a third relation
\begin{equation}\label{rightsum}
\sum_{i=1}^3\frac{pq}{(1-p-\lambda_i)(1+p-\lambda_i)}\alpha_i = 0.
\end{equation}
Equations (\ref{specdecomp1}), (\ref{leftsum}), and (\ref{rightsum}) form a linear system of $3$ equations with $3$ unknowns $(\alpha_1,\alpha_2,\alpha_3)$. It has the unique solution
 
\begin{equation}\label{rep13}
\alpha_1 = \frac{(\lambda_3-\lambda_2)(q-\lambda_1)(1+p-\lambda_1)}{(\lambda_2-\lambda_1)(q-\lambda_3)(1+p-\lambda_3)}\alpha_3,
\end{equation}
\begin{equation}\label{rep23}
\alpha_2 = \frac{(\lambda_1-\lambda_3)(q-\lambda_2)(1+p-\lambda_2)}{(\lambda_2-\lambda_1)(q-\lambda_3)(1+p-\lambda_3)}\alpha_3,
\end{equation}
\begin{equation}\label{form1}
\alpha_3 = \frac{(1+p-\lambda_3)(q-\lambda_3)}{(\lambda_3-\lambda_1)(\lambda_3-\lambda_2)}\alpha.
\end{equation}

It is possible to find $\alpha_3$ from $\alpha$ and $\lambda_3$ only. From (\ref{EVcubic}) we know that $R(\lambda_i)=\lambda$ for $i\in \{1,2,3\}$, which means that
\begin{equation}\label{zeqn}
(z-\lambda_1)(z-\lambda_2)(z-\lambda_3) = pqR(z) - \lambda.
\end{equation}
By differentiating both sides of (\ref{zeqn}) with respect to $z$, and then evaluating at $z=\lambda_3$, we obtain
\begin{equation}
(\lambda_3-\lambda_1)(\lambda_3-\lambda_2) = pqR'(\lambda_3) 
=3\lambda_3^2-6\lambda_3+2+pq.
\end{equation}
This allows us to replace the denominator in the RHS of (\ref{form1}), which leads to
\begin{equation}\label{decompform}
\alpha_3 = \frac{(1+p-\lambda_3)(q-\lambda_3)}{3\lambda_3^2-6\lambda_3+2+pq}\alpha.
\end{equation}
This proves the induction from $\alpha_{n,*}$ to $\alpha_{n+1,*}$ in the case where the eigenvalue $\lambda \notin \{0,2\}$.

If $\lambda \in \{0,2\}$, then $f$ has 2 eigenfunction extensions to the next level. The matching condition stated in (\ref{wsplit}) remains the same, but degenerates to a linear system of 2 equations with 2 unknowns. We omit the details.

%
%

As a simple corollary, we now show that the weights $\alpha_{*,*}$ are all nonnegative in our convention. Recall that $\alpha_{0,1} = \alpha_{0,2} = \frac{1}{2}$. By the structure of the cubic polynomial $R(z)$ (see Figure \ref{figCD}), $\lambda_1 \in [0, \min(p,q)]$, $\lambda_2 \in [\max(p,q),\min(1+p,1+q)]$, and $\lambda_3 \in [\max(1+p,1+q),2]$. So if $\alpha$ is nonnegative, it is direct to verify using (\ref{rep13}) through (\ref{form1}) that $\alpha_1$, $\alpha_2$, and $\alpha_3$ are all nonnegative. By induction we deduce that all weights $\alpha_{*,*}$ are nonnegative.
\end{proof}

Using the aforementioned result, we now define the ``approximate delta functions.''  Based on Lemma \ref{prop17}, and the fact that the lowest branch of $R^{-1}(z)$ is increasing, we can deduce that the lowest $|V_n|$ eigenvalues of $\Delta_{n+1}$ are determined recursively by
\begin{equation}\label{eigvrec}
\lambda_{n+1,k} = (\text{The lowest branch of~}R^{-1})(\lambda_{n,k})\quad \text{for}~0\leq k \leq 3^n.
\end{equation}
Given the level-$n$ delta function $\delta_0^{(n)}$, we define its approximation at level $n_0 < n$ by
\begin{equation}
\delta_0^{(n_0,n)}(x) := \sum_{k=0}^{3^{n_0}} \alpha_{n_0,k} f_{n,k}(x)\quad \text{for}~x\in V_n.
\end{equation}
In other words, we consider a truncated series of the spectral representation at level $n_0$, fixing the coefficients $\alpha_{n_0,k}$, but taking the eigenfunctions $f_{n,k}$ to level $n$.
 

\section{Estimates of eigenvalues and
eigenfunctions }\label{section5estimates}
In this section, we use the spectral decimation to derive finer estimates of the eigenvalues and the eigenfunctions, which will be used in Section \ref{sectionwaveequation}.
Of particular importance is the constant $\C := R'(0) = \frac{2+pq}{pq}$, the renormalization factor for the eigenvalues $\{\lambda_{n,k}\}_n$. Its significance derives from the following fact.

\begin{proposition}
\label{prop:renormeigv}
For each $k \in \mathbb{N}\cup\{0\}$, the limit $\lim_{n\to\infty} [R'(0)]^n \lambda_{n,k}$ exists.
\end{proposition}
\begin{proof}
Let $\varphi$ be the lowest branch of $R^{-1}$, which we regard as a function on $\mathbb{C}$. Via a power series expansion, we see that $\varphi(z)$ has an attracting fixed point at $z=0$, with $\varphi(0)=0$ and $\varphi'(0)= [R'(0)]^{-1}= \frac{pq}{2+pq}<1$. By Koenigs' theorem (see for example \cite[\S 8]{Milnor}), the renormalized iterates $\{z\mapsto [\varphi'(0)]^{-n} \varphi^n(z)\}_n$ converge uniformly on compact subsets of a local neighborhood of $0$.

Now given a fixed $k$, the recurrence relation (\ref{eigvrec}) implies that there exists $n_0=n_0(k)$ such that $\lambda_{n+1,k}= \varphi(\lambda_{n,k})$ for all $n\geq n_0$. Combine this with the foregoing result and we conclude that the limit
\begin{equation}
\lim_{n\to\infty} R'(0)^n \lambda_{n,k}=[\varphi'(0)]^{-n_0} \lim_{n\to\infty} [\varphi'(0)]^{-n+n_0}\varphi^{n-n_0}(\lambda_{n_0,k})
\end{equation}
exists.
\end{proof}

In what follows we denote $\lambda_k := \lim_{n\to\infty} \C^n \lambda_{n,k}$. Our next result gives an upper and a lower bound on $\lambda_k$.
\begin{theorem}
\label{thm:eigv}
Fix $p \in (0,\frac{1}{2})$, and let $k$ and $n_0$ be as in the proof of Proposition \ref{prop:renormeigv}.
Then
\begin{equation}
\label{eigvest}
 \lambda_{n_0,k} \left(1+\frac{3(pq)}{(2+pq)^2}\lambda_{n_0,k}\right)\le \frac{\lambda_k}{\C^{n_0}} \le \lambda_{n_0,k} \exp\left(\lambda_{n_0,k}\frac{p(2+q)}{2q(2-p)}\right).
\end{equation} 
\end{theorem}

\begin{proof}
As in the previous proof, let $\varphi$ be the lowest branch of $R^{-1}$. The lower bound on $\lambda_k$ will come from the Taylor approximation to $\varphi$, while the upper bound will come from a quadratic function which is at least as large as $\varphi$.

\emph{Lower bound.} We compute the Taylor series expansion of $\varphi$ about $0$ to 2nd order in $z$. It is
\begin{equation}\label{e-az}
a(z) = \frac{pq}{2+pq}z +\frac{3(pq)^2}{(2+pq)^3} z^2.
\end{equation}
This is explained by the fact that 
the first  derivative of the inverse function $\varphi(z)$ is given by $\varphi'=1/R'$, 
and its second derivative  is given by $\varphi''=-R''/(R')^3$. 
Computing these derivatives at zero gives 
the quadratic function \eqref{e-az}. 

Furthermore   we claim that $a(z) < \varphi(z)$ for all $z \in (0,2)$. It is enough to check that $\varphi'''(z) > 0$.  Here we use the identity 
\begin{equation}
\frac{d^3y}{dx^3} =-\frac{d^3x}{dy^3}\left(\frac{dy}{dx}\right)^4+3\left(\frac{d^2x}{dy^2}\right)^2\left(\frac{dy}{dx}\right)^5, 
\end{equation}
which in our context reads
\begin{equation}\label{varphi}
\varphi'''(z) = -R'''(\varphi(z)) [\varphi'(z)]^4 + 3[R''(\varphi(z))]^2 [\varphi'(z)]^5.
 \end{equation}
Since $\varphi'(z)>0$, we can factor out $[\varphi'(z)]^5$ from (\ref{varphi}), and use the identity $\varphi'(z) = [R'(\varphi(z))]^{-1}$ so that we reduce the original sign question to checking the sign of
\begin{eqnarray}
&& -R'''(\varphi(z)) R'(\varphi(z)) + 3 [R''(\varphi(z))]^2  \\
&=& \left(\frac{1}{pq}\right)^4 \left[-6\cdot(3[\varphi(z)]^2-6\varphi(z)+(2+pq))+3 (6\varphi(z)-6)^2\right] \\
&=& \left(\frac{1}{pq}\right)^4 \cdot 6 \cdot (15 [\varphi(z)]^2-30 \varphi(z)+16-pq)\\
&=& \left(\frac{1}{pq}\right)^4 \cdot 6 \cdot \left[15 (\varphi(z)-1)^2 +(1-pq)\right],
\end{eqnarray}
which is always positive. This shows that $\varphi'''(z) >0$, and thus $a(z) < \varphi(z)$ for $z\in (0,2)$. Combined with the fact that the functions $z\mapsto a(z)$ and $z\mapsto \varphi(z)$ are both monotone increasing on $[0,2)$, this implies that for each $n\in \mathbb{N}$, $a^n(z) \le \varphi^n(z)$ for $z\in [0,2)$.  

Fix $k\in \mathbb{N}\cup \{0\}$ and $n_0=n_0(k)$ as in the proof of Proposition \ref{prop:renormeigv}. Put $z_0=\lambda_{n_0,k}$, and define the sequence of numbers $\{z_n\}_n$ inductively by $z_{n+1} = \varphi(z_n)$. Then, 
by the inequalities above,  
we have
\begin{align}
\varphi(z_n) \ge a(z_n) &= \frac{pq}{2+pq}z_n +\frac{3(pq)^2}{(2+pq)^3} z_n^2 \\
&= z_n\left( \frac{pq}{2+pq} \right) \left(1+\frac{3(pq)}{(2+pq)^2}z_n\right) \\
&\ge a(z_{n-1})\left( \frac{pq}{2+pq} \right) \left(1+\frac{3(pq)}{(2+pq)^2}a(z_{n-1})\right).
\end{align}
Iterating this process we arrive at the estimate
\begin{equation}
\varphi(z_n) \geq z_{n_0}\left(\frac{pq}{2+pq}\right)^{n-n_0} \prod^n_{j={n_0}}\left(1+\frac{3(pq)}{(2+pq)^2}z_j\right).
\end{equation}
Noting that $1+\frac{3(pq)}{(2+pq)^2}z > 1$, we obtain a slightly crude but still efficient estimate
\begin{equation}
\varphi(z_n) \geq z_{n_0}\left(\frac{pq}{2+pq}\right)^{n-n_0} \left(1+\frac{3(pq)}{(2+pq)^2}z_{n_0}\right),
\end{equation}
which is the claimed lower bound in (\ref{eigvest}).

\emph{Upper bound.} To bound $\varphi(z)$ from above, we construct a quadratic function $h(z)$ such that $h(0)=0$, $h(2)=p$ and $h'(0)=\varphi'(0)$. A simple calculation shows that
\begin{equation}
h(z) = \frac{pqz}{2+pq}\left(1+\frac{p(2+q)}{4q}z\right),
\end{equation}
and $h(z) \geq \varphi(z)$ for $z\in [0,2]$. Using this, along with the fact that $z_{n+1} = \varphi(z_n) \le \frac{pq}{2+pq}z_n\left(1+\frac{p(2+q)}{4q} z_n\right)$, and $z_{n+1} \le \frac{p}{2}z_n$, we get the following estimate:
\begin{eqnarray}
\varphi(z_n) &\le& z_{n_0} \left(\frac{pq}{2+pq}\right)^{n-n_0}\prod^{n}_{j={n_0}}\left(1+\frac{p(2+q)}{4q} z_j\right) \\
&\le& \lambda_{n_0,k} \left(\frac{pq}{2+pq}\right)^{n-n_0}\prod^{n}_{j={n_0}}\left(1+\frac{p(2+q)}{4q} \left(\frac{p}{2}\right)^{j-n_0} \lambda_{n_0,k}\right).
\end{eqnarray} 
Therefore, using  the inequality $1+x\leq e^x$, 
\begin{eqnarray}
\nonumber \lambda_k &=& \lim_{n \to \infty} \left(\frac{2+pq}{pq}\right)^{n}\varphi(z_n) \\
\nonumber &\le& \lim_{n \to \infty} \left(\frac{2+pq}{pq}\right)^{n}\left(\frac{pq}{2+pq}\right)^{n-n_0}\lambda_{n_0,k}\prod^{n}_{j={n_0}}\left(1+\frac{p(2+q)}{4q} \left(\frac{p}{2}\right)^{j-n_0} \lambda_{n_0,k}\right) \\
\nonumber &\le& \left(\frac{2+pq}{pq}\right)^{n_0}\lambda_{n_0,k} \exp\left(\frac{p(2+q)}{4q}\frac{1}{1-\frac{p}{2}} \lambda_{n_0,k} \right) \\
&=& \left(\frac{2+pq}{pq}\right)^{n_0}\lambda_{n_0,k} \exp\left(\frac{p(2+q)}{2q}\frac{1}{2-p}\lambda_{n_0,k}\right).
\end{eqnarray}
This gives the claimed upper bound in (\ref{eigvest}).
\end{proof}

For a function $f: V_n \to\mathbb{R}$, we denote its sup norm by $\|f\|_{n,\infty} = \sup\{|f(x)|: x\in V_n\}$. Likewise, the sup norm of $h: I\to\mathbb{R}$ is denoted by $\|h\|_\infty = \sup\{|h(x)| : x\in I\}$. Our next result is an estimate on the sup norms of the eigenfunctions of $\Delta_n$.

\begin{lemma}
\label{lem:eigf}
Fix $p\in (0,\frac{1}{2})$. Let $f_{n,k}$ be the eigenfunction corresponding to the $(k+1)$th lowest eigenvalue $\lambda_{n,k}$ of $\Delta_n$. Then for every $m>n \geq n_0$,
\begin{align}
\label{eigf1}
\| f_{m,k} \|_{m,\infty} \le  \|f_{n,k}\|_{n,\infty} \prod^m_{j= n+1}\left(1+\left(\frac{p}{2}\right)^{j-n}\frac{\lambda_{n,k}}{q-\lambda_{n+1,k}}\right).
\end{align}
In particular, if $f_k = \lim_{i\to\infty} f_{n_0+i,k}$ per Theorem \ref{thm:eigfconti}, then
\begin{equation}
\|f_k\|_\infty \leq  \|f_{n,k}\|_{n,\infty} \exp\left(\frac{\lambda_{n,k}}{q-\lambda_{n+1,k}}\frac{p}{2-p} \right).
\label{eigf2}
\end{equation}

\end{lemma}

\begin{proof}
Let us introduce the function
\begin{equation}
F(A,B,z) = \frac{q(1-z)A+pqB}{(q-z)(1+p-z)},
\end{equation}
which is derived from the eigenfunction extension algorithm (\ref{ext0}) and (\ref{ext1}). Note that $z\in [0,p)$ because the extension uses the lowest branch of $R^{-1}$ starting from level $n_0$. First we would like to control the linear growth of $z\mapsto F(A,B,z)-F(A,B,0)$:
\begin{align}
& |F(A,B,z) - F(A,B,0)| \nonumber \\ \le& |A| \left|\frac{q(1-z)}{(q-z)(1+p-z)} - \frac{q}{1-p^2}\right| + |B|\left|\frac{pq}{(q-z)(1+p-z)} - \frac{pq}{1-p^2} \right| \nonumber \\
\le& \max(|A|,|B|) q \left( \left| \frac{(1-z)(1-p^2) - 1+2z+p^2-z^2}{(q-z)(1+p-z)(1-p^2)} \right| + p \left| \frac{2z-z^2}{(q-z)(1+p-z)(1-p^2)} \right|\right) \nonumber\\
 =& \max(|A|,|B|)q|z| \left( \left| \frac{1+p^2 - z}{(q-z)(1+p-z)(1-p^2)}\right| + p\left|\frac{2-z}{(q-z)(1+p-z)(1-p^2)}\right| \right). \label{Fz}
\end{align}
Since $0 < z < p < \frac{1}{2}$, the absolute value terms in the RHS of (\ref{Fz}) are positive, so we can drop the absolute value signs and add the two terms in the bracket to get
\begin{align}
& |F(A,B,z) - F(A,B,0)| \nonumber \\ \le& \max(|A|,|B|)q|z| \left( \frac{1+2p+p^2 - z - pz}{(q-z)(1+p-z)(1-p^2)} \right) \nonumber\\
=& \max(|A|,|B|)q|z| (1+p) \left( \frac{1+p-z}{(q-z)(1+p-z)(1-p^2)} \right) \nonumber\\
=& \max(|A|,|B|)\frac{|z|}{(q-z)}, \label{Fdifference}
\end{align}
which implies that
\begin{equation}
|F(A,B,z)| \le \max(|A|,|B|)|z|  \frac{1}{(q-z)} + |F(A,B,0)|. \label{Fmax}
\end{equation}


We use (\ref{Fmax}) to estimate the sup norms of the eigenfunctions: for all $n\geq n_0$,
\begin{eqnarray}
\|f_{n+1,k}\|_{n+1,\infty} &\leq& \|f_{n,k}\|_{n,\infty} |F(1,1,z_{n+1})| \nonumber \\
&\leq& \|f_{n,k}\|_{n,\infty} \left(\frac{ |z_{n+1}|  }{q-z_{n+1}} + |F(1,1,0)| \right) \nonumber \\
&\leq& \|f_{n,k}\|_{n,\infty} \left(\frac{z_{n+1}}{q-z_{n+1}}  + 1\right), \label{fineq}
\end{eqnarray}
Iterating the inequality (\ref{fineq}) and using the fact that $z_{n+m} \le \left(\frac{p}{2}\right)^m z_n$ (\ref{evineq}) gives
\begin{align}
\|f_{m,k}\|_{m,\infty} &\le \|f_{n,k}\|_{n,\infty} \prod^m_{j= n+1}\left(1+\frac{z_j}{q-z_j}\right) \nonumber\\
&\le \|f_{n,k}\|_{n,\infty} \prod^m_{j= n+1}\left(1+\left(\frac{p}{2}\right)^{j-n}\frac{z_n}{q-z_{n+1}}\right) \nonumber\\
&= \|f_{n,k}\|_{n,\infty} \prod^m_{j= n+1}\left(1+\left(\frac{p}{2}\right)^{j-n}\frac{\lambda_{n,k}}{q-\lambda_{n+1,k}}\right) \label{supfm}
\end{align}
for all $m>n \geq n_0$. This shows (\ref{eigf1}).

Next, using the triangle inequality and taking the supremum, we have
\begin{equation}
\left|\sup_{x\in I} |f_k(x)| - \sup_{x\in I} |f_{m,k}(x)|\right| \leq \sup_{x\in I} |f_k(x)-f_{m,k}(x)|.
\label{fkbound}
\end{equation}
Recall from Theorem \ref{thm:eigfconti} that the limit $\lim_{i\to\infty} f_{n_0+i,k}$ is uniform on $I$. This along with the bound (\ref{fkbound}) implies that
\begin{equation}
\|f_k\|_\infty = \lim_{m\to\infty}\|f_{m,k}\|_{m,\infty}
\end{equation}
So by taking the limit $m\to\infty$ on both sides of (\ref{lem:eigf}), we arrive at the estimate
\begin{eqnarray}
\|f_k\|_\infty &\leq &\|f_{n,k}\|_{n,\infty} \lim_{m\to\infty} \prod^m_{j= n+1}\left(1+\left(\frac{p}{2}\right)^{j-n}\frac{\lambda_{n,k}}{q-\lambda_{n+1,k}}\right) \nonumber \\
&\leq& \|f_{n,k}\|_{n,\infty} \exp\left(\sum_{j=n+1}^\infty\left(\frac{p}{2}\right)^{j-n}\frac{\lambda_{n,k}}{q-\lambda_{n+1,k}} \right)\nonumber \\
&=& \|f_{n,k}\|_{n,\infty} \exp\left(\frac{\lambda_{n,k}}{q-\lambda_{n+1,k}}\frac{p}{2-p} \right),
\end{eqnarray}
where in the second line we used the inequality $1+x\leq e^x$. This proves (\ref{eigf2}).
\end{proof}

The following result provides a quantitative estimate of the convergence of $f_{n,k}$ to $f_k$ in sup norm. As in the previous section, we harmonically extend $f_{n,k}$ from $V_n$ to $I$, and abuse notation by calling the extension $f_{n,k}$ still. Then $\Delta_\mu f_{n,k}(x)=0$ for all $x\in I\setminus V_n$.

\begin{theorem}\label{tfkbound}  
Fix $p\in (0,\frac{1}{2})$.
Then
\begin{align}
\| f_k-f_{n,k}\|_\infty \le  \C^{-n}\lambda_k \|f_{n,k}\|_{n,\infty} \,\|g\|_\infty \exp\left(\frac{\lambda_{n,k}}{q-\lambda_{n+1,k}}\frac{p}{2-p}\right), \label{ineq:fkinftynorm}
\end{align}
where $g: I\times I \to\mathbb{R}_+$ is the Green's function associated with $\Delta_\mu$, and
\begin{equation}
\|g\|_\infty := \sup_{(x,y)\in I\times I} \,g(x,y).
\end{equation}
\end{theorem}

\begin{proof}
Since $f_k$ and $f_{n,k}$ agree on $V_n$, it is enough to estimate their difference on $I\setminus V_n$. Based on the construction described in Section \ref{section2eigenvalues}, $I\setminus V_n$ is the disjoint union of $\{F_w\{(0,1)\}: |w|=n\}$, where $w=w_1 w_2\cdots w_n$ is a word of length $n$ with $w_i\in \{1,2,3\}$ for $1\leq i\leq n$, and $F_w:= F_{w_1} \circ F_{w_2}\circ \cdots \circ F_{w_n}$. So our task is to show that for every word $w$ of length $n$, the function $|(f_k-f_{n,k})\circ F_w|$ on $(0,1)$ has a uniform upper bound.

Our strategy is to exploit the self-similarity of the fractal Laplacian $\Delta_\mu$ (Proposition \ref{prop:Lapss}), as well as properties of the corresponding Green's function. We remind the reader that $G_\mu = \Delta_\mu^{-1}$ is the Green's operator associated to $\Delta_\mu$. It admits an integral kernel $g(\cdot,\cdot)$ called the Green's function, defined by
\begin{equation}
(G_\mu u)(y) = \int_I \,g(y,y')u(y') \,d\mu(y').
\label{eq:Green}
\end{equation}
The existence of $G_\mu$ and $g$ follows from the theory of Kigami \cite[\S3.5-\S3.6]{Ki}. In particular, $G_\mu : L^1(I,\mu) \to C(I)$, and $g$ is a nonnegative continuous function on $I\times I$.

To begin the proof, we start with the self-similarity of $\Delta_\mu$ (Proposition \ref{prop:Lapss}):
\begin{equation}
\Delta_\mu ((f_k-f_{n,k})\circ F_w) = \C^{-|w|} \Delta_\mu(f_k-f_{n,k})\circ F_w \quad \text{on} ~I.
\label{eq:sslap}
\end{equation}
Since $f_{n,k}$ is harmonic with respect to $\Delta_\mu$ on $F_w\{(0,1)\}$ for every word $w$ of length $n$, it follows that
\begin{equation}
\Delta_\mu(f_k -f_{n,k})\circ F_w = \Delta_\mu f_k \circ F_w  = \lambda_k f_k \circ F_w \quad \text{on}~(0,1).
\label{eq:reduce}
\end{equation}
Combine (\ref{eq:sslap}) and (\ref{eq:reduce}) and we get
\begin{equation}
\Delta_\mu((f_k-f_{n,k})\circ F_w) = \C^{-n} \lambda_k f_k \circ F_w \quad \text{on}~(0,1).\label{eq:formeq}
\end{equation}
Now apply the Green's operator $G_\mu$ on both sides of (\ref{eq:formeq}) to get
\begin{equation}
(f_k-f_{n,k})\circ F_w = \C^{-n} \lambda_k G_\mu(f_k \circ F_w) \quad \text{on}~(0,1).
\end{equation}
Using the representation (\ref{eq:Green}) we obtain the estimate
\begin{align}
|(f_k - f_{n,k})\circ F_w | &\leq \C^{-n} \lambda_k \|G_\mu(f_k\circ F_w)\|_\infty \nonumber \\
&\leq \C^{-n} \lambda_k \|g\|_\infty \|(f_k\circ F_w)\|_\infty \nonumber \\
&\leq \C^{-n} \lambda_k \|g\|_\infty \|f_k\|_\infty \quad \text{on}~(0,1).
\end{align}
This proves that
\begin{equation}
\|f_k- f_{n,k}\|_\infty \leq  {\bf C}_0^{-n}\lambda_k \|g\|_\infty \|f_k\|_\infty. \label{fkdiffest}
\end{equation}

In order to turn the RHS of (\ref{fkdiffest}) into a more useful estimate, we use (\ref{eigf2}) to replace $\|f_k\|_\infty$ by $\|f_{n,k}\|_{n,\infty}\exp\left(\frac{\lambda_{n,k}}{q-\lambda_{n+1,k}} \frac{p}{2-p}\right)$, which yields (\ref{ineq:fkinftynorm}).
%
%
\end{proof}

\section{Estimates on the solution of the wave equation}\label{sectionwaveequation} 

We now apply the results of Sections \ref{section3eigenfunctions} and \ref{section5estimates} to estimate the solution of the wave equation on the interval $I$ endowed with the fractal measure $\mu$.

Numerically, we can only compute the eigensolutions of the fractal Laplacian up to a finite level, so in practice we solve the ``approximate'' wave equation
\begin{equation}
\label{eq:approxwaveeqn}
\left\{ \begin{array}{ll} \partial_{tt} u_n = -\Delta_n u_n & \text{on}~ V_n \times [0,T],\\
u_n(\cdot,0) = \delta_0^{(n_0,n)} &\text{on}~V_n,\\
\partial_t u_n(\cdot,0)=0 & \text{on} ~V_n,
\end{array}\right.
\end{equation}
where $\displaystyle \delta_0^{(n,n_0)}=\sum_{k=0}^{3^{n_0}} \alpha_k f_{n,k}$ is the approximate $\delta$-function built up from the first $|V_{n_0}|=(3^{n_0}+1)$ eigenfunctions of $\Delta_n$ (with $\Delta_n f_{n,k} = \lambda_{n,k} f_{n,k}$), and $\alpha_k := \alpha_{n_0,k}\geq 0$ are the coefficients found in Section \ref{sec:specdecom}. Throughout the section $n_0$ will be fixed, and we will not mention $n_0$ explicitly unless the context demands it.

Following the exact same argument, the solution to (\ref{eq:approxwaveeqn}) has the series representation
\begin{equation}
\label{eq:unseries}
u_n(x,t) = \sum_{k=0}^{3^{n_0}} \alpha_k f_{n,k}(x) \cos\left(t\sqrt{\lambda_{n,k}}\right) \quad\text{for all}~x\in V_n ~\text{and}~t\in [0,T].
\end{equation}
For each $t$, we harmonically extend the function $x\mapsto u_n(x, t)$ from $V_n$ to $I$. This procedure allows us to compare $u_n(x,t)$ with
\begin{equation}
\label{eq:useries}
\tilde{u}(x,t) = \sum_{k=0}^{3^{n_0}} \alpha_k f_k(x) \cos\left(t\sqrt{\lambda_k}\right) \quad\text{for all}~ x\in I ~\text{and}~t\in [0,T],
\end{equation}
the solution of the wave equation on $(I,\mu)$ whose initial condition is the truncated series representation of the $\delta$-impulse.
We note that $\tilde{u}$ is 
differentiable in $t$ and continuous in $x$ because the eigenfunctions functions $f_k$ are continuous. 
However it is highly localized function 
at $t=0$, and therefore it 
mimics wave propagation from a delta function initial values.

\subsection{Convergence of approximate  solutions of wave equation.}

In this subsection we establish an upper bound on
\begin{equation}
\label{eq:wavediff}
\left|u_n \left(x,\frac{t\pi}{\sqrt{\lambda_{n,1}}}\right) - \tilde{u}\left(x,\frac{t\pi}{\sqrt{\lambda_1}}\right)\right|
\end{equation}
for all $n$ and $t$, uniform in $x$. This would then give us the convergence of $u_n$ to $\tilde{u}$ at fixed $t$ and uniformly in $x$. Note that we are normalizing $t$ in such a way that the orthogonal projection of the wave onto the lowest eigenfunction (corresponding to eigenvalue $\lambda_{*,0}$) propagates at speed $1$.

\begin{theorem}
Fix $p\in (0,\frac{1}{2})$. Let $u_n$ and $\tilde u$ be respectively defined as in (\ref{eq:unseries}) and (\ref{eq:useries}). Then there exists a positive constant $C=C(n_0, p)$ such that for each $t\in [0,T]$ and $n> n_0$,
\begin{equation}
\sup_{x\in I} \left|u_n \left(x,\frac{t\pi}{\sqrt{\lambda_{n,1}}}\right) - \tilde{u}\left(x,\frac{t\pi}{\sqrt{\lambda_1}}\right)\right| \leq C\left(t\vee 1\right)\C^{-n}.
\label{ineq:wavediff}
\end{equation}
\end{theorem}

\begin{proof}
Using the series representations (\ref{eq:unseries}) and (\ref{eq:useries}) and the triangle inequality, we find that (\ref{eq:wavediff}) is bounded from above by 
\begin{equation}
\sum_{k=0}^{3^{n_0}} \alpha_k\left|f_{n,k}(x) \cos \left(t\pi \frac{\sqrt{\lambda_{n,k}}}{\sqrt{\lambda_{n,1}}}\right) -f_k(x) \cos \left(t\pi \frac{\sqrt{\lambda_k}}{\sqrt{\lambda_1}}\right) \right|,
\end{equation}
which, by a simple manipulation using the sum-to-product trigonometric rules, is equal to
\begin{eqnarray}
&& \sum_{k=0}^{3^{n_0}} \alpha_k \left|[f_{n,k}(x) - f_k(x)] \cos\left(t\pi \Lambda_{n,k}^+ \right) \cos\left( t\pi \Lambda_{n,k}^-\right)\right. \nonumber \\
&& \qquad \qquad \left. - [f_{n,k}(x) + f_k(x)] \sin \left(t\pi \Lambda_{n,k}^+\right) \sin\left(t\pi \Lambda_{n,k}^-\right)\right|,\label{eq:waveest2}
\end{eqnarray}
where $\Lambda_{n,k}^{\pm} := \frac{1}{2} \left(\sqrt{\frac{\lambda_{n,k}}{\lambda_{n,1}}} \pm \sqrt{\frac{\lambda_k}{\lambda_1}} \right)$.  Using again the triangle inequality, we can estimate (\ref{eq:waveest2}) from above by $I_1 + I_2$, where
\begin{eqnarray}
\label{eq:I1} I_1 &:=& \sum_{k=0}^{3^{n_0}} \alpha_k \left|f_{n,k}(x) - f_k(x)\right| \cdot\left|\cos\left(t\pi \Lambda_{n,k}^+ \right) \right|\left|\cos\left( t\pi \Lambda_{n,k}^-\right)\right|,\\
\label{eq:I2} I_2 &:=& \sum_{k=0}^{3^{n_0}} \alpha_k \left|f_{n,k}(x) + f_k(x)\right| \cdot \left|\sin \left(t\pi \Lambda_{n,k}^+\right) \right| \left|\sin\left(t\pi \Lambda_{n,k}^-\right)\right|.
\end{eqnarray}
The key term to control in $I_1$ is $|f_{n,k}(x) - f_k(x)|$, while in $I_2$ it is $\left|\sin\left(t\pi \Lambda_{n,k}^-\right)\right|$. For the former we invoke Theorem \ref{tfkbound},
while for the latter we apply the Taylor expansion
\begin{equation}
\label{eq:sinest}
\left|\sin\left(t\pi \Lambda_{n,k}^-\right)\right| \leq t\pi \left|\Lambda_{n,k}^-\right| + O\left((t\Lambda_{n,k}^-)^3\right).
\end{equation}
For terms other than these two, we apply the simple minded estimates $
\left|\cos\left(t\pi \Lambda_{n,k}^{\pm}\right)\right|\leq 1$, $\left|\sin\left(t\pi \Lambda_{n,k}^+\right)\right|\leq 1$, and 
\begin{align}
\left|f_{n,k}(x) + f_k(x)\right| &\leq 2|f_{n,k}(x)| + \|f_{n,k}- f_k\|_\infty \nonumber \\
&\leq 2\|f_{n,k}\|_\infty + O\left(\C^{-n}\right).
\end{align}

First let us estimate $\Lambda_{n,k}^-$, which amounts to controlling the ratio $\frac{\lambda_{n,k}}{\lambda_{n,1}}$. By Theorem \ref{thm:eigv},
\begin{equation}
\label{ineq:eigv}
 \frac{\left(1+\frac{3pq}{(2+pq)^2}\lambda_{n,1}\right)} {\exp\left(\lambda_{n,k} \frac{p(2+q)}{2q(2-p)}\right)}\leq \frac{\lambda_{n,k}}{\lambda_{n,1}}\cdot \frac{\lambda_1}{\lambda_k} \leq \frac{\exp\left(\lambda_{n,1} \frac{p(2+q)}{2q(2-p)}\right)}{\left(1+\frac{3pq}{(2+pq)^2}\lambda_{n,k}\right)}.
\end{equation}
We note that in this case we have 
$ 0 \leqslant k \leqslant 3^{n_0} <n$ and 
so Theorem~\ref{thm:eigv} is applicable.

In what follows we denote $\D := \frac{p(2+q)}{2q(2-p)}$ and $\Di := \frac{3pq}{(2+pq)^2}$. Observe from the discussion in the proof of Theorem \ref{thm:eigv} that $\D > \Di$ whenever $p\in (0,1]$.

From Proposition \ref{prop:renormeigv}, we know that for each fixed $k$, $\lambda_{n,k}= O(\C^{-n})$ as $n\to\infty$. Thus upon expanding the LHS and the RHS of (\ref{ineq:eigv}) up to the $O(\C^{-n})$ terms, we get 
\begin{align}
 1 + \frac{1}{2}\left( \Di\lambda_{n,1} -  \D \lambda_{n,k}\right) +& o(\C^{-n}) \leq \sqrt{\frac{\lambda_{n,k}}{\lambda_{n,1}} \cdot \frac{\lambda_1}{\lambda_k}}  \nonumber \\
 \leq& ~1 + \frac{1}{2} \left(\D\lambda_{n,1} -  \Di \lambda_{n,k}\right) + o(\C^{-n}).
\end{align}
 It then follows that
\begin{equation}
\frac{1}{4}\left(\Di \lambda_{n,1} -\D\lambda_{n,k}\right) + o(\C^{-n}) \leq \Lambda_{n,k}^- \cdot \sqrt{\frac{\lambda_1}{\lambda_k}} \leq \frac{1}{4}\left(\D\lambda_{n,1} - \Di\lambda_{n,k}\right) +o(\C^{-n}).
\end{equation}
Plugging this into (\ref{eq:sinest}) yields
\begin{eqnarray}
&& \left|\sin\left(t\pi \Lambda_{n,k}^-\right)\right| \nonumber \\&\leq& \frac{t\pi }{4} \sqrt{\frac{\lambda_k}{\lambda_1}} \max\left(\left|\D \lambda_{n,1}-\Di\lambda_{n,k}\right|, \left|\D\lambda_{n,k} - \Di \lambda_{n,1}\right|\right) + o(t\C^{-n})\nonumber \\
&=& \frac{t\pi}{4} \sqrt{\frac{\lambda_k}{\lambda_1}} \left|\D \lambda_{n,k} - \Di\lambda_{n,1}\right| + o(t\C^{-n}).
\end{eqnarray}
Putting all the estimates into (\ref{eq:I1}) and (\ref{eq:I2}) gives
\begin{align}
I_1 &\leq \C^{-n} \|g\|_\infty\sum_{k=0}^{3^{n_0}} \alpha_k \lambda_k  \|f_{n,k}\|_{n,\infty} \exp\left(\frac{\lambda_{n,k}}{q-\lambda_{n+1,k}} \frac{2}{2-p}\right)   +  o\left(\C^{-n}\right).\nonumber \\
I_2 &\leq \frac{t\pi}{2} \sum_{k=0}^{3^{n_0}} \alpha_k  \|f_{n,k}\|_\infty \sqrt{\frac{\lambda_k}{\lambda_1}} \left|\D \lambda_{n,k}-\Di\lambda_{n,1}\right| + o(t\C^{-n}).
\end{align}
This means that (\ref{eq:wavediff}) is bounded above by
\begin{align}
\sum_{k=0}^{3^{n_0}} \alpha_k \|f_{n,k}\|_\infty \left[\lambda_k \C^{-n} \|g\|_\infty \exp\left(\frac{\lambda_{n,k}}{q-\lambda_{n+1,k}} \frac{2}{2-p}\right)\right. \nonumber \\ \left.+ \frac{t\pi}{2} \sqrt{\frac{\lambda_k}{\lambda_1}} \left|\D\lambda_{n,k} - \Di \lambda_{n,1}\right|\right] + o\left(\left(t\vee 1\right) \C^{-n}\right).
\label{eq:finalest}
\end{align}

It remains to explain how (\ref{ineq:wavediff}) follows from (\ref{eq:finalest}).
For $0\leq k \leq 3^{n_0}$, $\lambda_k \leq \lambda_{3^{n_0}}$. Also by Proposition \ref{prop:renormeigv}, there exists a positive constant $C$ independent of $n$ such that
\begin{equation}
|\D \lambda_{n,k} - \Di \lambda_{n,1}| \leq C \C^{-n} |\D \lambda_k- \Di \lambda_1|.
\end{equation}


By an argument of Kigami \cite{Ki}, the Green's function corresponding to $\Delta_\mu$ on $I$ can be constructed independently of the measure $\mu$, whence independently of the value $p$. In particular when $p=\frac{1}{2}$, we recover the classical Green's function on $I$ with Lebesgue measure, $g(x,y) = (x\wedge y)((1-x) \wedge (1-y))$. Thus $\|g\|_\infty = \frac{1}{4}$.

For the exponential in the first term, note that since $p\in (0,\frac{1}{2})$,
\begin{equation}
\frac{\lambda_{n,k}}{q-\lambda_{n+1,k}} \frac{2}{2-p} < \frac{\lambda_{n,k}}{\frac{1}{2}-\lambda_{n+1,k}} \frac{2}{2-\frac{1}{2}} = \frac{4}{3} \frac{\lambda_{n,k}}{\frac{1}{2}-\lambda_{n+1,k}}.
\label{eq:expest}
\end{equation}
By Proposition \ref{prop:renormeigv}, the RHS of (\ref{eq:expest}) is $O(\C^{-n})$ as $n\to\infty$. Thus the exponential is $\exp(O(\C^{-n}))=1+O(\C^{-n})$.

\begin{figure}[htb]
\begin{center}
\includegraphics[scale = 0.2]{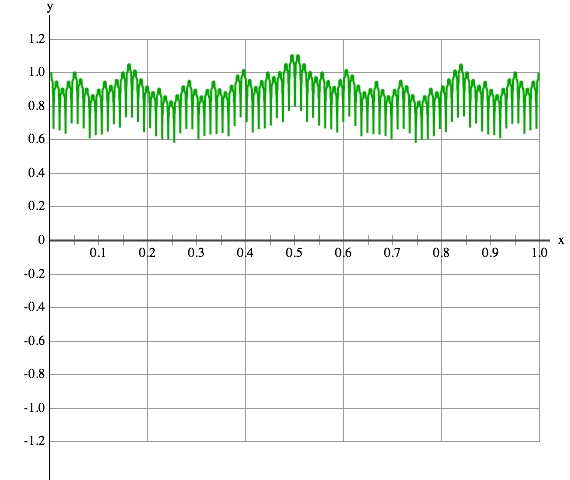}
\includegraphics[scale = 0.2]{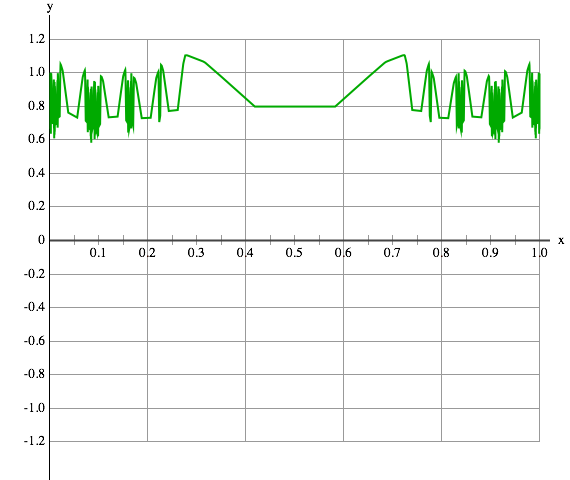}
\includegraphics[scale = 0.2]{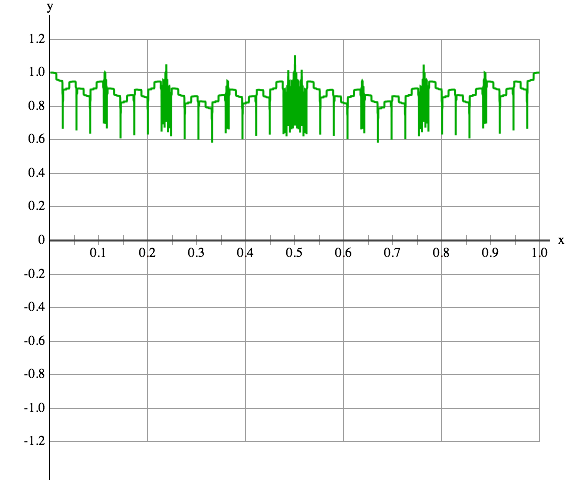}
\caption{A graph of $y(x) = \sum^{3^{n_0}}_{k=0} \alpha_k |f_{n,k}(x)|$ and in order: uniform spacing of points, uniform resistance between points, and uniform measure.}
\label{figure6}
\end{center}
\end{figure}

Finally, we claim that there exists a constant $C$ such that
$$\sum_{k=0}^{3^{n_0}} \alpha_k \|f_{n,k}\|_\infty \leq C$$
for all $n > n_0$. To explain this, note that $\sum_{k=0}^{3^{n_0}} \alpha_k =1$, by virtue of our choice of the initial eigenfunctions ($f_{0,1}$ and $f_{0,2}$) and initial weights $\alpha_{0,1}=\alpha_{0,2}=\frac{1}{2}$, and the matching condition (\ref{specdecomp1}). Then 
the rest of the proof follows from Section~8 of L.~Rogers' paper \cite{Ro}. 
Numerically, $C$ is slightly above $1$, see Figure \ref{figure6}. 
\end{proof}

\section{Numerical computation of 
eigenfunctions and solutions of the 
wave equation}\label{section4numerical}
We present some of the numerical results obtained by   our spectral decimation method.  The spectral decimation is an iterative method, and the code repeats the calculation done in section \ref{section3eigenfunctions}. This code, which is used to produce pictures and to perform the experiments, and a graphical user interface to recreate the results can be found at \url{http://homepages.uconn.edu/fractals/fractalwave/}.  
Here we give a representative variety of figures detailing some of the numerical calculations that have been performed.  Figures \ref{f1}, and \ref{f2} show the first 25 eigenfunctions, and, in particular, the ways in which the symmetry is broken for values of $p$ and $1-p$.  Figures \ref{fdelta1}, and \ref{fdelta2} show the quality of the approximation for the delta function for various values of $n_0$.  In particular one can see that for small values of $p$ (in our case $p = .2$) the approximation is significantly better than for the corresponding values  of $1-p$ (in our case $q = .8$).  This shows why our efforts focused on the cases where $p < .5$.  The next set of figures (\ref{figure1},   \ref{figure2}, \ref{figure3},   \ref{figure4}) heuristically suggest that the visible portion of the wave propagates at a speed proportional to $t^\frac{d_s}{2}$.  But, further investigation will be needed to show this more precisely. 
Figures~\ref{figure5} 
shows three different parametrizations of 
a representative eigenfunction.

\begin{center}
\begin{figure}[!]
\caption{Uniform spacing of points at time $t = 0.4$.  From left to right $p = 0.1, 0.2, 0.3, 0.4$.}
\includegraphics[width = 0.4\textwidth]{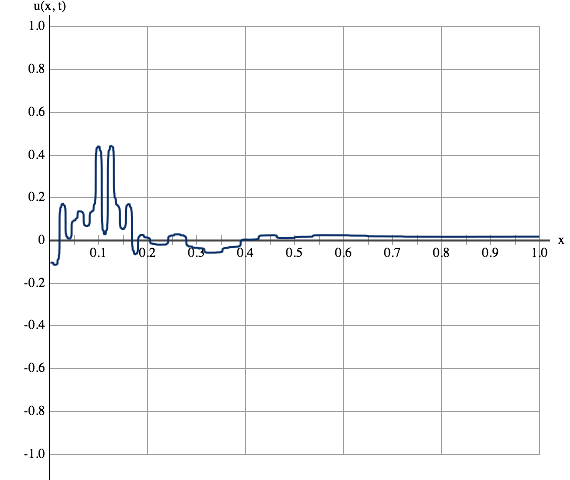}
\includegraphics[width = 0.4\textwidth]{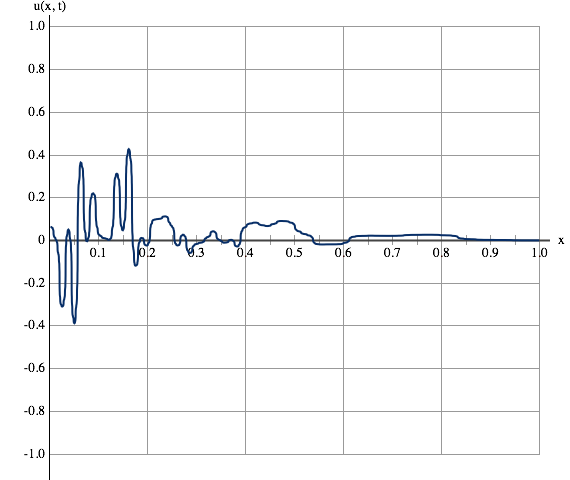}\\
\includegraphics[width = 0.4\textwidth]{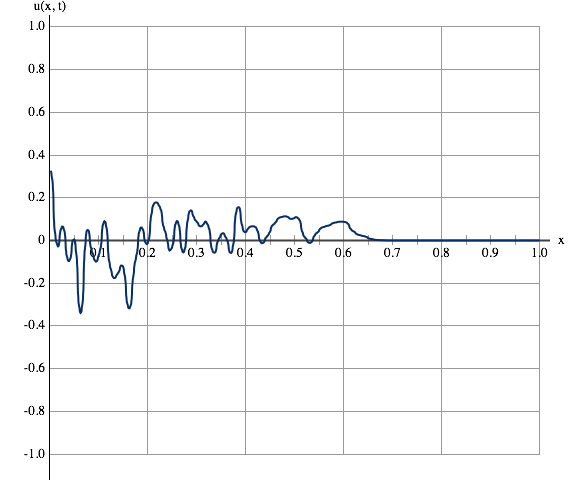}
\includegraphics[width = 0.4\textwidth]{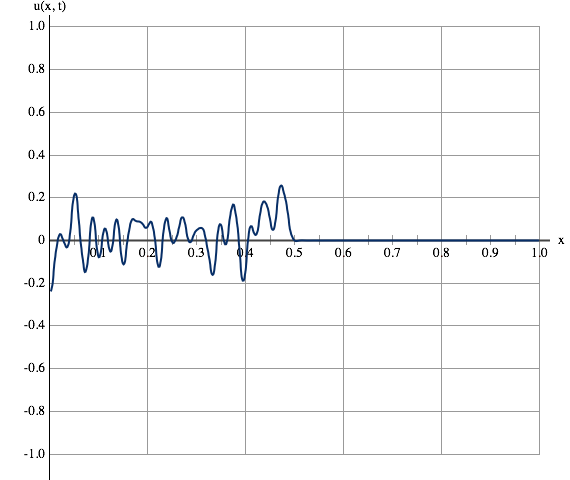}
\label{figure1}
\end{figure}
\end{center}

\begin{center}
\begin{figure}[!]
\caption{Uniform spacing of points at time $t = 0.1$.  From left to right $p = 0.1, 0.2, 0.3, 0.4$.}
\includegraphics[width = 0.4\textwidth]{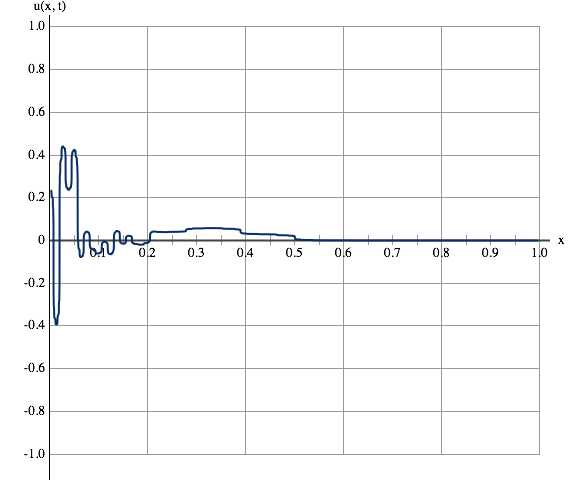}
\includegraphics[width = 0.4\textwidth]{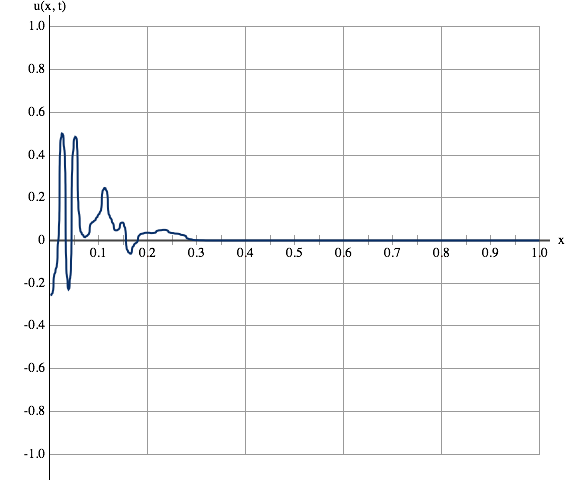}\\
\includegraphics[width = 0.4\textwidth]{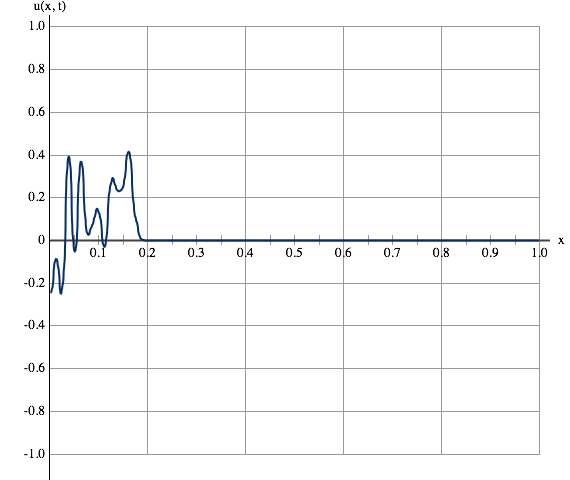}
\includegraphics[width = 0.4\textwidth]{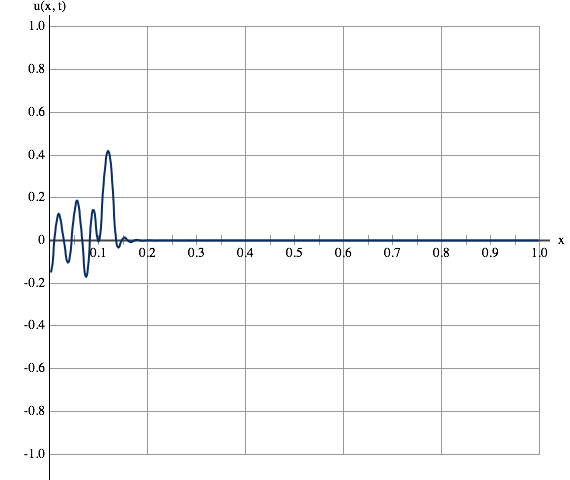}
\label{figure2}
\end{figure}
\end{center}


\begin{center}
\begin{figure}[!]
\caption{Uniform spacing of points and, from left to right: time t = 0.1, 0.2, 0.3, 0.4. with p = 0.1}
\includegraphics[width = 0.4\textwidth]{t1p1uni}
\includegraphics[width = 0.4\textwidth]{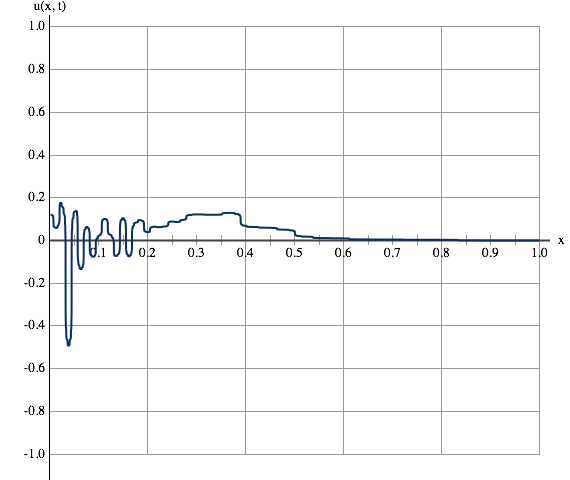}\\
\includegraphics[width = 0.4\textwidth]{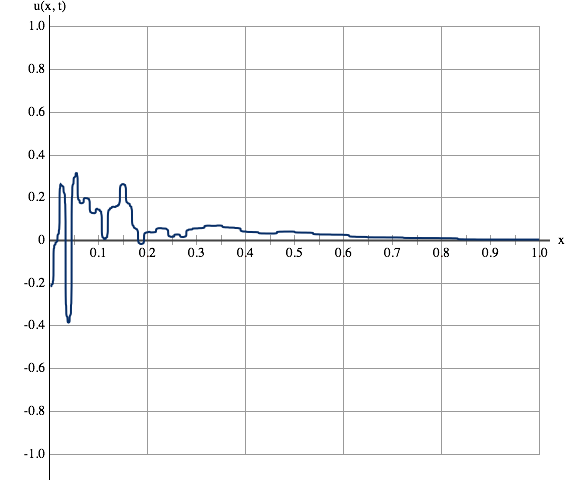}
\includegraphics[width = 0.4\textwidth]{t4p1uni}
\label{figure3}
\end{figure}
\end{center}

\begin{center}
\begin{figure}[!]
\caption{Time $t = 0.4$ and $p = 0.1$ and   from left to right: uniform spacing of points, uniform resistance between points, and uniform measure.}
\includegraphics[scale = 0.20]{t4p1uni}
\includegraphics[scale = 0.20]{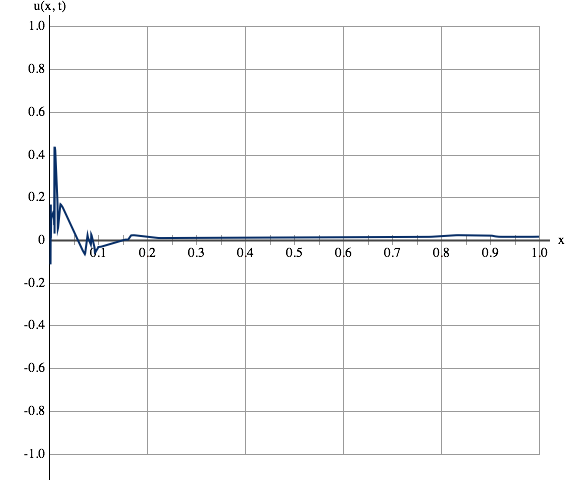}
\includegraphics[scale = 0.20]{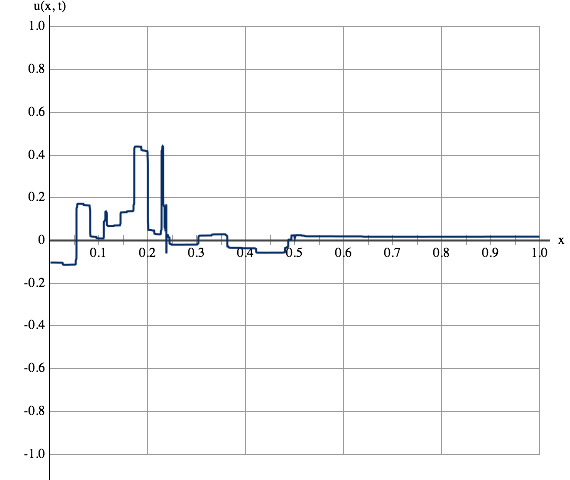}
\label{figure4}
\end{figure}
\end{center}

\begin{center}
\begin{figure}[!]
\caption{Eigenfunction 53, with $p=0.1$ and in order: uniform spacing of points, uniform resistance between points, and uniform measure.}
\includegraphics[scale = 0.20]{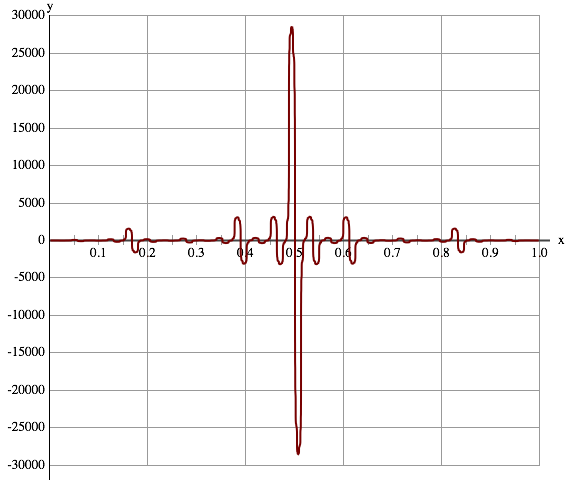}
\includegraphics[scale = 0.20]{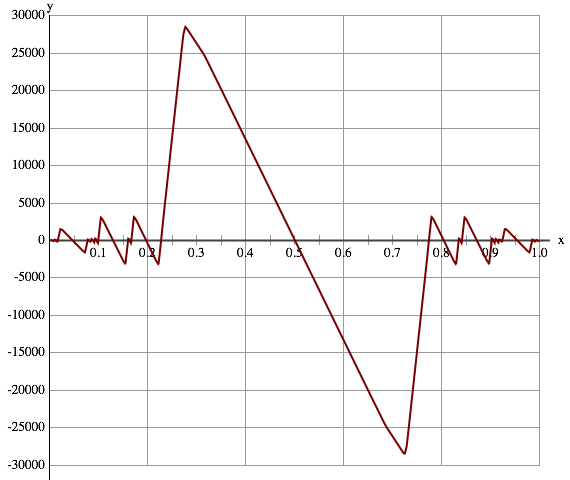}
\includegraphics[scale = 0.20]{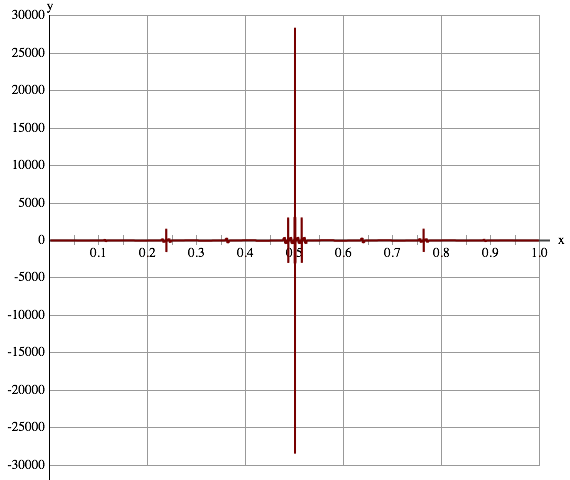}
\label{figure5}
\end{figure}
\end{center}


\end{document}